\documentclass[a4paper,twocolumn,11pt]{quantumarticle}
\pdfoutput=1
\usepackage{graphicx}
\usepackage{epsfig}
\usepackage{amsfonts}
\usepackage{amsmath}
\usepackage{amssymb}
\usepackage{amsthm}

\usepackage{color}
\usepackage[colorlinks=true,linkcolor=blue,citecolor=magenta,urlcolor=blue]{hyperref}
\usepackage{comment}
\usepackage{braket}
\usepackage{dsfont}

\newtheorem{corrolary}{Corrolary}
\newtheorem{definition}{Definition}

\newtheorem{proposition}{Proposition}

\newtheorem{remark}{Remark}
\newtheorem{theorem}{Theorem}

\newcommand{\tr}{{\rm tr}}

\usepackage{tikz}
\usetikzlibrary{shapes.geometric}
\usetikzlibrary{intersections,calc}
\usepackage{caption}
\usepackage{subcaption}
\usepackage{float}

\usepackage{enumerate}

\begin{document}


\title{Perfect quantum strategies with small input cardinality}


\author{Stefan Trandafir}
\email{strandafir@us.es}
\affiliation{Departamento de F\'{\i}sica Aplicada II, Universidad de Sevilla, E-41012 Sevilla, Spain}

\author{Junior R. Gonzales-Ureta}
\affiliation{Department of Physics, Stockholm University, 10691 Stockholm, Sweden}

\author{Ad\'an Cabello}
\email{adan@us.es}
\affiliation{Departamento de F\'{\i}sica Aplicada II, Universidad de Sevilla, E-41012 Sevilla, Spain}
\affiliation{Instituto Carlos~I de F\'{\i}sica Te\'orica y Computacional, Universidad de Sevilla, E-41012 Sevilla, Spain}

\maketitle


\begin{abstract}
A perfect strategy is one that allows the mutually in-communicated players of a nonlocal game to win every trial of the game. Perfect strategies are basic tools for some fundamental results in quantum computation and crucial resources for some applications in quantum information. 
Here, we address the problem of producing qudit-qudit perfect quantum strategies with a small number of settings. For that, we exploit a recent result showing that any perfect quantum strategy induces a Kochen-Specker set. We identify a family of KS sets in even dimension $d \ge 6$ that, for many dimensions, require the smallest number of orthogonal bases known: $d+1$. 
This family was only defined for some $d$. We first extend the family to infinitely many more dimensions.
Then, we show the optimal way to use each of these sets to produce a bipartite perfect strategy with minimum input cardinality. As a result, we present a family of perfect quantum strategies in any $(2,d-1,d)$ Bell scenario, with $d = 2^kp^m$ for $p$ prime, $m \geq k \geq 0$ (excluding $m=k=0$), $d = 8p$ for $p \geq 19$, $d=kp$ for $p > ((k-2)2^{k-2})^2$ whenever there exists a Hadamard matrix of order $k$, other sporadic examples, as well as a recursive construction that produces perfect quantum strategies for infinitely many dimensions $d$ from any dimension $d'$ with a perfect quantum strategy. We identify their associated Bell inequalities and prove that they are not tight, which provides a second counterexample to a conjecture of 2007.
\end{abstract}


\section{Introduction}


\subsection{Motivation}


The qubit is the basic unit of quantum information \cite{Schumacher:1995PRA}. Bell nonlocality \cite{Bell:1964PHY,Clauser:1969PRL} is arguably the most powerful resource of quantum theory. Interestingly, ``maximal'' bipartite Bell nonlocality {\em cannot} be achieved with qubit-qubit correlations, but requires qudit-qudit correlations with $d \ge 3$.

There are several ways to quantify nonlocality. However, it has recently been shown \cite{Liu:2023XXX} that nonlocality is maximal at the same time in, at least, four of them: (I) When nonlocality is quantified through the nonlocal fraction \cite{Elitzur:1992PLA}. In this case, maximum nonlocality corresponds to nonlocal fraction $1$. Then, the correlation is said to be {\em fully nonlocal} (FN) \cite{Aolita:2012PRA}. (II) When nonlocality is quantified by the minimum distance to a face of the nonsingnaling polytope without local points \cite{Goh2018PRA}. In this case, maximum nonlocality occurs when the correlation is in a face of the nonsingnaling polytope without local points. Then, the correlation is said to be {\em face nonsignaling} (FNS) \cite{Liu:2023XXX}. (III) When nonlocality is quantified based on whether or not it allows for a GHZ-like proof \cite{GHZ89} of Bell theorem. In this case, maximum nonlocality occurs when the correlation allows for such a proof. Then, the correlation is said to allow an {\em ``all-versus-nothing''} (AVN) \cite{Cabello:2001PRLb} proof of nonlocality. (IV) When nonlocality is quantified by the wining probability of a nonlocal game whose classical winning probability is smaller. In this case, maximum nonlocality occurs when the correlation allows the players to always win the game. Then, the correlation is said to allow for a {\em perfect quantum strategy} (PQS) \cite{CHTW04} or {\em ``quantum pseudo-telepathy''} \cite{GBT05,BroadbentPhD2008}. For precise definitions of FN, FNS, and AVN, we refer the reader to Appendix~\ref{app.Correlationstrength}. A detailed definition of PQS is presented below. In this paper, we will focus and use the language of PQSs. However, our results can also be read in terms of FN correlations, FNS correlations, and AVN proofs. 

In the bipartite case, a nonlocal game \cite{Aravind:2004AJP} $G = (X \times Y, A \times B, \pi, W)$ is characterized by: (i) Two input sets: $X$ for the first player, Alice, and $Y$ for the second player, Bob. (ii) Two output sets: $A$ for Alice and $B$ for Bob. (iii) A distribution of probability for the inputs: $\pi(X \times Y)$. (iv) The winning condition, i.e., the condition that inputs and outputs must satisfy for winning the game: $W(X \times Y \times A \times B) \in \{0,1\}$, where $W=1$ means winning and $W=0$ means losing. Therefore, the winning probability of the game is
\begin{equation}
\label{win}
\omega(G) = \sum_{x,y,a,b} \pi(x,y) p(a,b|x,y) W(a,b,x,y).
\end{equation} 
The game admits a PQS if there is a quantum correlation $p(a,b|x,y)$, where $a \in A$, $b\in B$, $x \in X$, and $y \in Y$, such that $\omega(G)=1$. 

PQSs have a special status in foundations of quantum computation and play a crucial role in the proofs of some fundamental results such as the quantum computational advantage for shallow circuits \cite{Bravyi:2018SCI}, MIP$^\ast$=RE \cite{Ji:2021CACM}, and the impossibility of classically simulating quantum correlations with arbitrary relaxations of measurement and parameter independence \cite{Vieira;2024XXX}.

For $n \ge 3$, $n$-partite PQSs can be achieved with qubits \cite{Mermin:1990PT} and there are systematic methods to produce them for any number of parties \cite{PhysRevLett.95.120405,Cabello:2008PRA}. However, for $n=2$, PQSs require pairs of qudits with $d \ge3$ \cite{Brassard:2005}. 
Bipartite PQSs are therefore a fundamental motivation for exploring high-dimensional quantum correlations (other motivations are faster data rates, improved communication security \cite{BechmannPasquinucci:2000PRL}, higher resistance to noise \cite{Cerf:PRL2002}, and lower detection efficiencies to attain the loophole-free regime \cite{ZSSLPC2023}).

On the other hand, bipartite PQSs require more than two measurements per party \cite{Gisin:2007IJQI} and more than two outcomes per measurement (at least for some measurements) \cite{CHTW04}. 

Three types of bipartite PQSs are known:
\begin{enumerate}
    \item Bipartite PQSs based on qudit-qudit maximally entangled states and Kochen-Specker (KS) \cite{Kochen:1967JMM} sets of rays in dimension $d \ge 3$ \cite{Stairs:1983PS,HR83,Brown:1990FPH,Elby:1992PLA,Renner2004b,CHTW04,Aolita:2012PRA}. KS sets are defined in Sec.~\ref{sec.N=7}.
    \item Bipartite PQSs based on magic sets (as defined in \cite{Arkhipov:2012XXX,Trandafir:2022PRL}) of Pauli observables \cite{Cabello:2001PRLb,Mancinska:2007} or parallel repetitions of them \cite{Coladangelo2016arxiv,Coudron2016arxiv,Araujo:2020Quantum}. Magic sets are related to KS sets \cite{Peres:1991JPA}.
    \item Bipartite PQSs based on state-independent contextuality sets of vectors \cite{Kleinmann:2012PRL} whose corresponding orthogonality graph has fractional packing number equal to the Lov\'asz number and larger than the independence number \cite{ZSSLPC2023}. In light of the one-to-one connection between PQSs and KS sets \cite{Cabello:2023XXX} (see Appendix~\ref{app.connection}), these sets must be KS sets.
\end{enumerate}

The simplest known example of a bipartite PQS, and probably the simplest allowed by quantum theory
\cite{Cabello:2023XXX}, occurs in the $(2,3,4),$ Bell scenario, where there are $2$ parties and each of them has $3$ settings with $4$ outcomes. This PQS requires Alice and Bob to share a ququart-ququart maximally entangled state. It was introduced in \cite{Cabello:2001PRLa,Cabello:2001PRLb} and makes use of the Peres-Mermin two-qubit proof of the KS theorem \cite{Peres:1990PLA,Mermin:1990PRLb}. This PQS is usually called the ``magic square'' correlation \cite{Aravind:2004AJP} and has been experimentally tested using photonic hyperentanglement: either polarization-path hyperentanglement \cite{CinelliPRL2005,YangPRL2005,Aolita:2012PRA} or polarization-orbital angular momentum hyperentanglement \cite{Xu:2022PRL}. To our knowledge, no other bipartite PQS has been experimentally tested.

The aim of this work is to show that, for infinitely many (finite) {\em even} dimensions 
$d \ge 4$, there are qudit-qudit PQSs using only $d-1$ local settings per party. The case $d = 4$, the magic square, is well-known, but the new ones, which are not {\em directly} connected to the case $d = 4$, form an especially beautiful and seemingly fundamental family of PQSs.


\subsection{Structure}


In Sec.~\ref{sec.TheGame}, we introduce a family of bipartite nonlocal games, the ``colored odd-pointed star games''. What makes this family special is explained in Secs.~\ref{sec.InfiniteClass} and~\ref{sec.Optimality}.

In Sec.~\ref{sec.InfiniteClass}, we show that that, infinitely many of the ``stars'' used in these games can be associated to a KS set in dimension~$d$. The interest of these KS sets is triple: they are very symmetric (their graph of orthogonality is vertex transitive), the impossibility of a KS assignment is easy to check (as they allow for a ``parity proof'' KS set \cite{Cabello:1996PLA,LisonekPRA2014}), and they have a small number of bases ($d+1$, which implies that they are the KS sets with the smallest number of bases known in each dimension $d$). As we will see, this last property is crucial for our purpose of identifying bipartite perfect quantum strategies with small input cardinality.
However, this extended family of KS sets is interesting in its own right, as is the fact that there is a connection between some of its new members and the smallest KS set in quantum theory \cite{Cabello:1996PLA,Xu:2020PRL}, which happens to be in $d=4$. Moreover, this family seems to be the ``natural'' (although nontrivial) extension to any even $d$ of the smallest KS. For all these reasons, these results are presented in a section, Sec.~\ref{sec.InfiniteClass}, that can be read independently from the rest of the article.

In Sec.~\ref{sec.Optimality}, we prove that colored odd-pointed star games correspond to the optimal ways of distributing the KS sets of Sec.~\ref{sec.InfiniteClass} between two players, while achieving a PQS with minimal input cardinality. 

In Sec.~\ref{sec.BellIneqs}, we discuss the Bell inequalities corresponding to the PQSs introduced in Secs.~\ref{sec.TheGame} and~\ref{sec.Optimality} and prove that, unlike the Bell inequalities corresponding to the magic square \cite{Cabello:2001PRLb,Gisin:2007IJQI} and GHZ \cite{Mermin:1990PRLa} games, the Bell inequalities corresponding to the new PQSs are not tight. This provides a second counterexample \cite{Liu:2023XXX} to a conjecture formulated in \cite{Gisin:2007IJQI}.

In Sec.~\ref{sec.Comparison}, we discuss some advantages of the new PQSs with respect to existing PQSs. We also compute the necessary visibility to achieve quantum advantage and compare it the visibility needed for different variants of the colored odd-pointed star games.

Finally, in Sec.~\ref{sec:conc}
we summarize the results and list some open problems.


\section{The colored odd-pointed star games}\label{sec.TheGame}


\subsection{The game}


Here, we introduce a class of bipartite nonlocal games that, in infinitely many cases, allow for a PQS with quantum advantage. We call them \emph{colored odd-pointed star games}.


Let $N$ be an odd number. The referee gives Alice a number $x$ chosen from
\begin{equation}
    X := \{1,2, \dots, N-2\}
\end{equation}
and gives Bob a number $y$ chosen from
\begin{equation}
    Y := \{3, 4, \dots, N\}.
\end{equation}
Alice outputs a number $a \in A$ and Bob outputs a number $b \in B$, where
\begin{eqnarray}
    A &:=& \{1, 2, \dots, N\} \setminus \{x\}, \\
    B &:=& \{1, 2, \dots, N\} \setminus \{y\}.
\end{eqnarray}
See Fig.~\ref{fig:pair_game}. Alice and Bob win when the sets $\{x,a\}$ and $\{y,b\}$ are either the same or disjoint. That is, whenever
\begin{equation}
    |\{x,a\} \cap \{y,b\}| \in \{0,2\}.
\end{equation}
Therefore, the winning condition is
\begin{equation}
    W(x,y,a,b) \equiv 1 + |\{x,a\} \cap \{y,b\}| \mod 2.
\end{equation}


\begin{figure}
    \centering
    \includegraphics[scale=0.7]{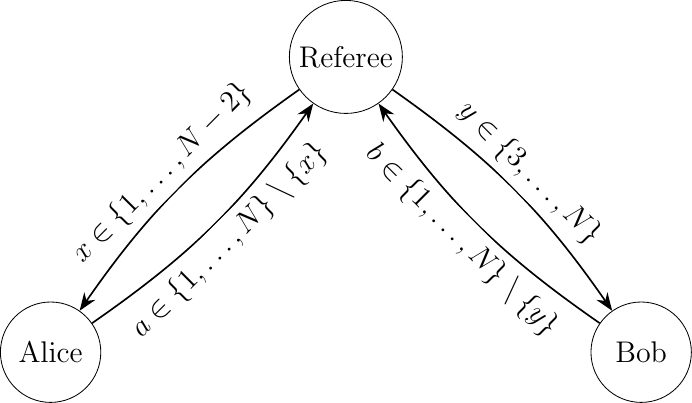}
    \caption{The colored odd-pointed star game. The referee provides Alice with a number $x$ between $1$ and $N-2$ and provides Bob with a number $y$ between $3$ and $N$. Alice (Bob) gives the referee back a number $a$ ($b)$ between $1$ and $N$ that is different from $x$ ($y$). Alice and Bob win whenever $\{a,x\}$ and $\{b,y\}$ are the same or disjoint.}
    \label{fig:pair_game}
\end{figure}


This game has no perfect classical strategy
(see Corollary \ref{cor.Optimal} in Sec.~\ref{sec.InfiniteClass}). However, as we will show in Sec.~\ref{sec.InfiniteClass}, there are infinitely many $N$ for which there is a PQS. 

The colored odd-pointed star game can be viewed as being played on a \emph{configuration}, which is a set of points $P$ and lines $L$ such that each point is incident (i.e., lies on) to the same number of lines, and each line is incident to the same number of points. The configuration that our game is played on is the $J(N,2)$-configuration with odd $N$ (or \emph{$N$-pointed star}). The name follows from the fact that, as we shall see in Sec.~\ref{sec.InfiniteClass}, it arises from the Johnson graph $J(N,2)$. The $J(N,2)$-configuration has $N$ lines labeled $\{1,2,\dots,N\}$ with $N-1$ points each. Each pair Line $i$ and Line $j$ meet at exactly one point which we label by $\{i,j\}$ (or $ij$ for short).  

Alice's (Bob's) inputs are {\em red} ({\em blue}) lines of the $J(N,2)$-configuration. Their outputs correspond to points on their respective lines. For example, if Alice is given (red) Line $5$, Alice's ``output $6$'' corresponds to point $\{5,6\}$ (i.e., the point in the intersection between Line $5$ and Line $6$). The cases $N=7$ and $N=9$ are illustrated in Figs.~\ref{fig:J72-bks} and \ref{fig:J92-bks}, respectively.

Alice and Bob win either if they output the same point or if they output points in different lines. That is, they lose if they output different points in the same line.

The reason for the $J(N,2)$-configuration with odd $N$ is that, as we will see in Sec.~\ref{sec.InfiniteClass}, for infinitely many values on $N$, it can be realized with a KS set in dimension $d=N-1$ and having only $N$ bases. The reason why some lines have only one color while other lines have the two colors is that, as we will see in Sec.~\ref{sec.Optimality}, this minimizes the input cardinality of the PQS. 


\begin{figure*}[htbp]
    \centering
    \begin{subfigure}[b]{\textwidth}
        \centering
         \includegraphics[scale=0.7]{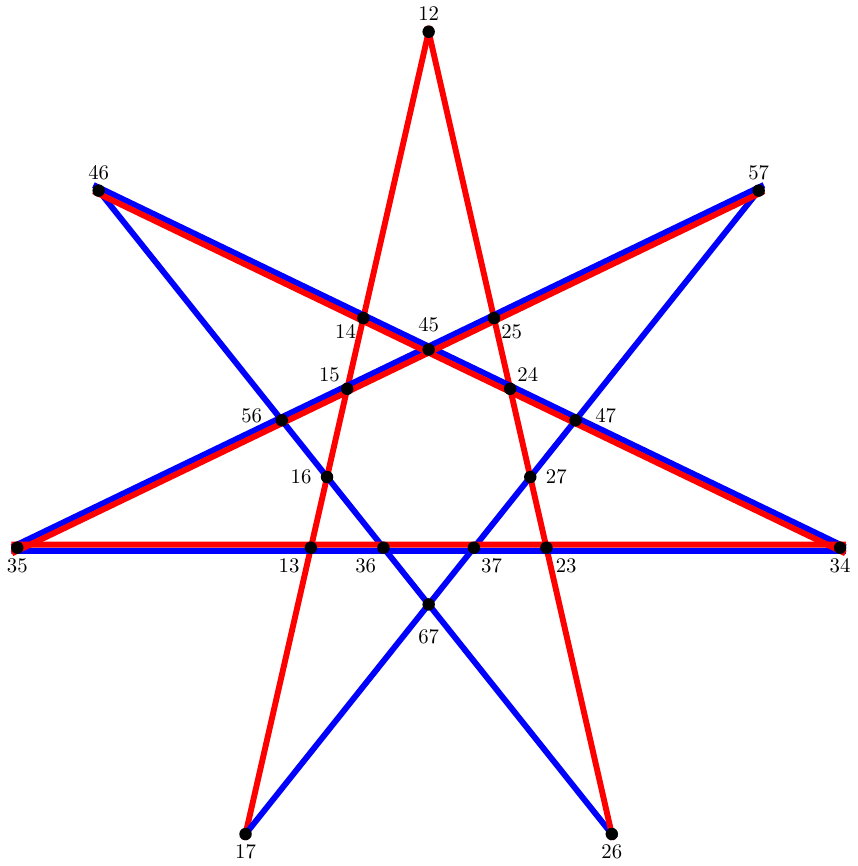}
         \caption{}
         \label{fig:J72-bks}
     \end{subfigure}
     
     \vspace{1cm}
     
     \begin{subfigure}[b]{\textwidth}
        \centering
         \includegraphics[scale=0.7]{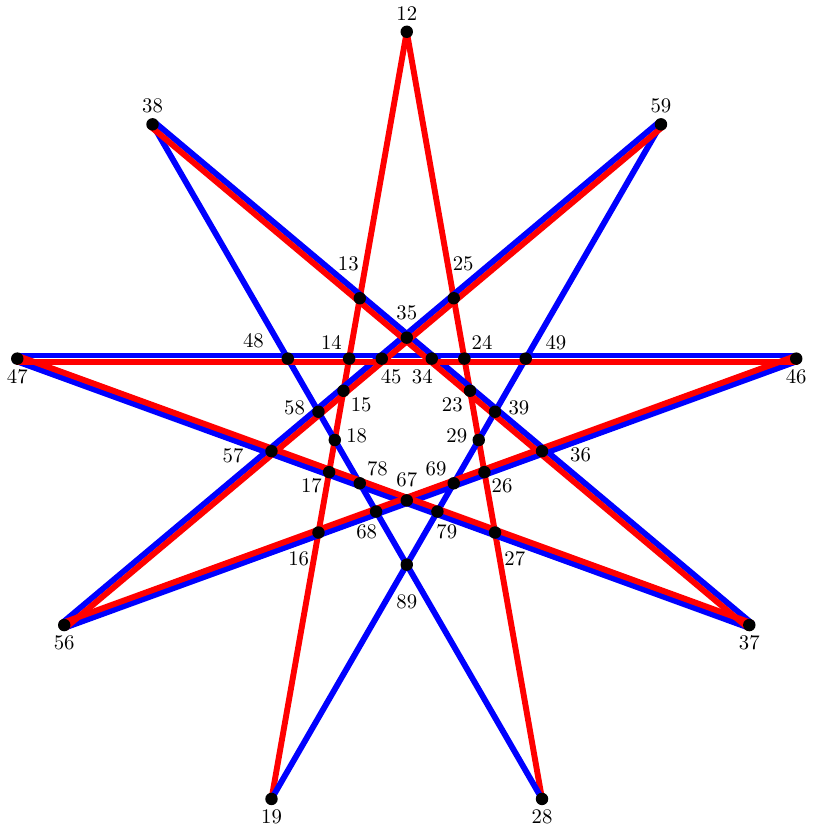}
         \caption{}
         \label{fig:J92-bks}
     \end{subfigure}
     
    \vspace{0.8cm}

     \caption{(a) Configuration associated to the colored odd-pointed star game with $N=7$. (b) Configuration associated to the colored odd-pointed star game with $N=9$. In both cases, each of the red lines $\{1,2,\dots,N-2\}$ corresponds to a possible input for Alice and each of the blue lines $\{3,4,\dots,N\}$ corresponds to a possible input for Bob. Black dots correspond to the possible outputs.}
\end{figure*}


\subsection{Example: N=7} \label{sec.N=7}

 
$N=7$ is the smallest instance of the colored odd-pointed star game where there is a PQS. The best classical strategy yields a winning probability of $24/25$. For example, an optimal classical strategy is the following: Alice gives the output corresponding to the point $12$ when the input corresponds to Line $1$ or to Line $2$, outputs $34$ for inputs $3$ or $4$, and outputs $56$ for input $5$. Bob outputs $34$ for inputs $3$ or $4$, outputs $56$ for inputs $5$ or $6$, and outputs $67$ for input $7$. The strategy is summarized in Table~\ref{tab:classical_7}. This strategy leads to a single failure: when Alice is given Line $5$ and Bob is given Line $7$.


\begin{table}[hbtp] 
    \centering
    \begin{tabular}{|c||c|c|c|c|c|c|c|}
        \hline 
        Line & 1 & 2 & 3 & 4 & 5 & 6 & 7 \\  \hline \hline 
        Alice & 12 & 12 & 34 & 34 & 56 &  &  \\
        Bob &  &  & 34 & 34 & 56 & 56 & 67 \\
        \hline 
    \end{tabular}
    \caption{An optimal classical strategy for the colored odd-pointed star game with $N=7$. For each input in the top row, Alice's and Bob's outputs are shown in the corresponding column.} 
    \label{tab:classical_7}
\end{table}


A PQS is the following. Alice and Bob share the following qudit-qudit maximally entangled state:
\begin{equation}
    \ket{\psi} = \frac{1}{\sqrt{6}}\sum_{k=0}^{5} \ket{k} \otimes \ket{k}.
\end{equation}
If Alice receives Line $j$ as input, she measures on her qudit the orthogonal basis $j$, defined as those vectors in Table~\ref{tab:raysJ72} having $j$ as superindex. This measurement projects Alice's qudit state into one state $v^{i,j}$ (or $v^{j,i}$). Then Alice outputs the number $i$. Similarly, if Bob receives Line $k$, he measures on his qudit the conjugate orthogonal basis $k$, defined as the conjugates of the vectors in Table~\ref{tab:raysJ72} having $k$ as superindex. This projects Bob's qudit state into a state $v^{i',k}$, and then Bob outputs the number $i'$.


\begin{table}
\centering
\begin{tabular}{cc}
    $v^{1, 2} = [0, 0, 0, 0, 0, 0]$ &
    $v^{1, 3} = [1, 2, 0, 2, 0, 1]$ \\
    $v^{1, 4} = [1, 0, 2, 2, 1, 0]$ &
    $v^{1, 5} = [0, 2, 2, 0, 1, 1]$\\
    $v^{1, 6} = [2, 2, 0, 1, 1, 0]$ & 
    $v^{1, 7} = [2, 0, 2, 1, 0, 1]$\\
    $v^{2, 3} = [2, 1, 0, 1, 0, 2]$ &
    $v^{2, 4} = [2, 0, 1, 1, 2, 0]$ \\
    $v^{2, 5} = [0, 1, 1, 0, 2, 2]$ &
    $v^{2, 6} = [1, 1, 0, 2, 2, 0]$ \\
    $v^{2, 7} = [1, 0, 1, 2, 0, 2]$ &
    $v^{3, 4} = [2, 2, 2, 1, 1, 1]$ \\
    $v^{3, 5} = [1, 1, 2, 2, 1, 2]$ &
    $v^{3, 6} = [0, 1, 0, 0, 1, 1]$ \\
    $v^{3, 7} = [0, 2, 2, 0, 0, 2]$ &
    $v^{4, 5} = [1, 2, 1, 2, 2, 1]$ \\
    $v^{4, 6} = [0, 2, 2, 0, 2, 0]$ &
    $v^{4, 7} = [0, 0, 1, 0, 1, 1]$\\
    $v^{5, 6} = [2, 1, 2, 1, 2, 1]$ &
    $v^{5, 7} = [2, 2, 1, 1, 1, 2]$ \\
    $v^{6, 7} = [1, 2, 2, 2, 1, 1]$
\end{tabular}
\caption{A quantum realization of $J(7,2)$ in dimension $d=6$. $[1, 2, 0, 2, 0, 1]$ is the abbreviation of the vector $(\zeta_3^1, \zeta_3^2, \zeta_3^0, \zeta_3^2, \zeta_3^0, \zeta_3^1)$, where $\zeta_3 = e^{\frac{2\pi i}{3}}$.}
\label{tab:raysJ72}
\end{table}


That a set of $21$ $6$-dimensional vectors exists such that each point in the $J(7,2)$ configuration can be assigned to a different vector and such that vectors in the same line are mutually orthogonal 
was presented in \cite{Lisonek2019}. We will refer to this set of vectors as a quantum realization of $J(7,2)$. Interestingly, this set is a KS set.

\begin{definition}[Kochen-Specker set]
\label{def:ks} 
A Kochen-Specker (KS) set is a finite set of vectors $\mathcal{V}$ in a Hilbert space of finite dimension $d \ge 3$, which does not admit an assignment $f: \mathcal{V} \rightarrow \{0,1\}$ satisfying: (I)~Two mutually orthogonal vectors $v_i$ and $v_j$ cannot have $f(v_i)=f(v_j)=1$. (II)~For every orthogonal basis (set
of $d$ mutually orthogonal vectors) $\{v_i\}_{i=1}^d$, $\sum_i f(v_i)=1$.
\end{definition}

It is sometimes useful to refer to a KS set as a pair $(\mathcal{V}, \mathcal{B})$, where $\mathcal{V}$ is the finite set of vectors and $\mathcal{B}$ is a specific subset of the orthogonal bases formed by elements of $\mathcal{V}$. 
The reason is that, in some cases, the no existence of $f : \mathcal{V} \to \{0,1\}$ satisfying (I) and (II) follows from the no existence of $f$ satisfying (II) for the orthogonal bases in $\mathcal{B}$.

The way to prove that $f$ does not exist for any quantum realization of $J(7,2)$ in dimension~$6$ is by noticing that there is an {\em odd} number ($7$) of lines (orthogonal bases). Therefore, there are $7$ equations of the type $\sum_i f(v_i)=1$. If we add them all, then the right-hand side is an odd number ($7$). However, this is impossible to achieve since each vector is in an {\em even} number (two) of lines. Sets that can be proven to be KS sets with a simple parity argument of this type are said to admit a \emph{parity proof}.

Interestingly, the KS set in Table~\ref{tab:raysJ72} is the KS set with the smallest number of vectors (and, more interestingly for us, bases) known in $d=6$. Moreover, it has been proven that it is the KS set with the smallest number of vectors in $d=6$ that admits a symmetric parity proof \cite{Lisonek2019}. The set is also symmetric in another sense: the graph $J(7,2)$ is vertex transitive (i.e., its automorphism group acts transitively on its vertices).


\section{An infinite family of symmetric KS sets with a small number of bases} \label{sec.InfiniteClass}


In this section we prove that, for infinitely many $N$, there are $N-1$-dimensional vectors such that each point in the $J(N,2)$-configuration can be assigned to a different vector such that vectors in the same line are mutually orthogonal. This implies that, for these values of $N$, there is a KS set in dimension $N-1$ and $N$ bases (and admitting a symmetric parity proof, and associated to a vertex-transitive graph). As we will see in Sec.~\ref{sec.Optimality}, this also implies that, for these values of $N$, there is a bipartite PQS with quantum advantage with $N-2$ inputs per party and $N-1$ outputs.


\subsection{Construction of the KS sets}


The problem of for which odd $N$ there is a quantum realization in dimension $N-1$ for the $J(N,2)$-configuration is a special case of a well-studied problem in graph theory - that of computing orthogonal representations. Given a graph $G=(V,E)$, where $V$ is the vertex set and $E$ is the edge set, an \emph{orthogonal representation} (OR) in $\mathbb{C}^d$ of $G$ is a labeling of the vertices of $G$ with non-zero vectors in $\mathbb{C}^d$ such that adjacent vertices are labeled by orthogonal vectors. Note that some papers use a different definition in which non-adjacent vertices are labeled by orthogonal vectors. The \emph{dimension} of the orthogonal representation is $d$, and the \emph{orthogonal rank} is the minimum dimension $d$ for which there is an orthogonal representation of $G$ in $\mathbb{C}^d$. The orthogonal rank is denoted by $\xi(G)$. For general graphs, computing orthogonal representations of minimal dimension is very difficult and little is known. In fact, there is not even an efficient algorithm to compute $\xi(G)$ (for example, deciding whether $\xi(G) \leq t$ for some positive integer $t \geq 3$ is NP-hard in the real case - i.e., substituting $\mathbb{C}$ for $\mathbb{R}$ - see \cite[Theorem 3.1]{Peeters:1996}). Bounds exist for $\xi(G)$: 
\begin{equation} \label{eq.xi_bounds}
    \omega(G) \leq \xi(G) \leq \chi(G),
\end{equation}
where $\omega(G)$ is the \emph{clique number} (maximum number of mutually adjacent vertices), and $\chi(G)$ is the \emph{chromatic number} (fewest colors needed in order to color each vertex so that adjacent vertices are assigned different colors). 
 
The \emph{Johnson graph} $J(N,m)$ is the graph whose vertices are the $m$-subsets of $\{1,\dots, N\}$, where an edge exists between vertices $S_1$ and $S_2$ if and only if $S_1$ and $S_2$ are not disjoint. Orthogonal representations of $J(N,2)$ in dimension $N-1$ are KS sets. Here, the bounds given in \eqref{eq.xi_bounds} are useful: for $J(N,2)$ with $N$ odd, $N \geq 5$, it is known that $\omega(G) = N-1$ and $\chi(G) = N$, so that the minimum dimension $d$ of any orthogonal representation must be one of $N-1$ or $N$. Note that in the $N=3$ case we have $\omega(G) = N$, so that there is no orthogonal representation in dimension $N-1$ (and thus no associated KS set). The previous discussion can be summarized as follows:
\begin{proposition}
For odd $N$, $\xi(J(N,2))~\in~\{N-1, N\}$.
    If $\xi(J(N,2)) = N-1$, then there is a KS set in dimension $N-1$.
\end{proposition}
Showing that $\xi(J(N,2)) = N-1$ amounts to proving that a KS exists. However, one still needs to explicitly construct the orthogonal representation in order to obtain the KS set.

In the next subsection we will provide a construction that works for infinitely many dimensions. In the subsection following that, we present a novel recursive construction that allows one to cover infinitely many dimensions from any known dimensions.


\subsection{The Jungnickel-Lison\v ek construction} \label{sec:jungnickel-lisonek}


For each dimension $d=2^{k}p^{m}$, with $m \geq k \geq 1$, it is possible to explicitly construct an orthogonal representation of $J(d+1,2)$ in $\mathbb{C}^d$.
The construction uses a result of Jungnickel \cite{Jungnickel79} in order to build a generalized Hadamard matrix $GH(2p, 2)$ and then produce the orthogonal representation via a pair of results of Lison\v ek \cite{Lisonek2019}.

A \emph{generalized Hadamard matrix} $GH(g, \lambda)$ over a group $G$ of order $g$ is a $g\lambda \times g\lambda$ matrix whose entries are elements of $G$ and for which each element of $G$ occurs exactly $\lambda$ times in the difference of any pair of rows. 
For example, the matrix
\begin{equation} \label{eg:gen-hadamard}
D^{(3)} :=
\begin{pmatrix}
0 & 0 & 0 & 0 & 0 & 0 \\
1 & 2 & 0 & 2 & 0 & 1 \\
1 & 0 & 2 & 2 & 1 & 0 \\
0 & 2 & 2 & 0 & 1 & 1 \\
2 & 2 & 0 & 1 & 1 & 0 \\
2 & 0 & 2 & 1 & 0 & 1
\end{pmatrix}
\end{equation}
is a generalized Hadamard matrix $GH(3,2)$ over $\mathbb{Z}_3$. The reader may notice a resemblance between the matrix $D^{(3)}$ and the KS set given in Table~\ref{tab:raysJ72}. As we shall see in this section, the KS set is constructed from the matrix.

By $EA(q)$ we denote the group of order $q$ obtained by the direct product of groups of prime order (for example, if $q = 12$, then $EA(12) = \mathbb{Z}_2 \oplus \mathbb{Z}_2 \oplus \mathbb{Z}_3$). The following theorem describes Jungnickel's construction of generalized Hadamard matrices.

\begin{theorem} \cite[Theorem 2.4]{Jungnickel79} \label{theo:Jungnickel}
    Let $q$ be an odd prime power and let $c$ be any non-square element of $GF(q)$. Fix an ordering $\sigma$ on the elements of $GF(q)$. Define matrices $A_i = (a^{i}_{\sigma(x), \sigma(y)})$ for $i=1, \dots, 4$ via
    \begin{align}
        a^{1}_{\sigma(x),\sigma(y)} &= xy + (x^2/4), \\ a^{2}_{\sigma(x), \sigma(y)} &= xy + (cx^2/4), \\
        a^3_{\sigma(x), \sigma(y)} &= xy - y^2 - (x^2/4), \\ a^4_{\sigma(x), \sigma(y)} &= [xy - y^2 - (x^2/4)]/c
    \end{align}
    and set 
    \begin{equation}
        D := \begin{pmatrix}
            A_1 & A_2 \\
            A_3 & A_4
        \end{pmatrix}.
    \end{equation}
    Then $D$ is a generalized Hadamard matrix $GH(q, 2)$ over $EA(q)$.
\end{theorem}

In the case that $q$ is prime, $EA(q) \simeq \mathbb{Z}_q$, so we obtain a generalized Hadamard matrix over $\mathbb{Z}_q$. For example the $GH(3,2)$ matrix $H$ given in Eq.~\eqref{eg:gen-hadamard} was obtained from Jungnickel's construction with $q=3$. The Kronecker product of generalized Hadamard matrices is also a generalized Hadamard matrix. Exploiting this observation and Jungnickel's construction, one immediately obtains generalized Hadamard matrices $GH(p, \lambda)$ over $\mathbb{Z}_p$ for each $n := p\lambda = 2^kp^m$ where $m \geq k \geq 1$ (see \cite[Corollary 2.5]{Jungnickel79}).

The Jungnickel-Lison\v ek construction for KS sets does not work directly from generalized Hadamard matrices, but from S-Hadamard matrices, which we now define. A matrix $H = (h_{i,j}) \in \mathbb{C}^{n \times n}$ is an \emph{S-Hadamard matrix of order $n$} if
\begin{enumerate}
    \item $HH^* = nI$ ,
    \item $|h_{i,j}| = 1$ for each $1 \leq i,j \leq n$, 
    \item $\sum_{j=1}^{n} h^2_{k,j}\overline{h^2_{\ell, j}} = 0$ for each $1 \leq k,\ell \leq n$ with $k \neq \ell$.
\end{enumerate}

The next proposition allows one to construct S-Hadamard matrices from generalized Hadamard matrices over $\mathbb{Z}_g$ (for $g \geq 3$). 

\begin{proposition} \cite[Proposition 2.4]{Lisonek2019}
    Let $g \geq 3$, and let $H = (h_{i,j})$ be a generalized Hadamard matrix $GH(g, \lambda)$ over $\mathbb{Z}_g$. Denote by $\zeta_g$ a primitive $g$th root of unity. Then the matrix $S = (\zeta_g^{h_{i,j}})$ is an S-Hadamard matrix.
\end{proposition}

In the case that $N := g\lambda$ is even, we obtain a KS set with at most $\binom{N+1}{2}$ vectors, and $N+1$ orthogonal bases via the construction described in \cite[Theorem 3.1]{Lisonek2019}. We describe this construction now, following the notation of \cite{Lisonek2019}.

Let $H = (h_{i,j})_{i,j=1}^{n}$ be an S-Hadamard matrix of order $n$. By $h_i$ we denote the $i$th row of $H$. If $h_1$ is not the all-ones vector, we normalize $H$ by multiplying each of its entries $h_{i,j}$ by $h_{1,j}^{-1}$.
For vectors $u = (u_1, \dots, u_n),v = (v_1, \dots, v_n) \in \mathbb{C}^n$, we write $u \circ v$ to denote the vector $(u_1v_1, \dots, u_nv_n)$. 

From $H$, we construct the KS set $(\mathcal{V}, \mathcal{B})$ where $\mathcal{V}$ is the set of vectors and $\mathcal{B}$ is the set of orthogonal bases. We index each vector by a pair of indices $i,j$ each between $1$ and $N+1$ - so the elements of $\mathcal{V}$ are each written in the form $v^{i,j}$. It should be understood that $v^{i,j}$ and $v^{j,i}$ denote the same vector.

The vectors are:
\begin{itemize}
    \item $v^{1,s} := h_{s-1}$ for each $s = 2, \dots, n+1$. 
    \item $v^{2,s} := h_{s-1} \circ h_{s-1}$ for each $s=3, \dots, n+1$.
    \item $v^{r,s} := h_{r-1} \circ h_{s-1}$ for each pair $r=3,\dots, n+1$, $s=r+1,\dots n+1$.
\end{itemize}
One can check that 
\begin{equation}
    B_r = \{v^{r,i} : i=1,\dots, n, i \neq r\}
\end{equation}
is an orthogonal basis for each $r=1, \dots, n+1$.  Finally, setting $\mathcal{B} = \{B_1, \dots, B_{n+1}\}$, the previously described parity argument shows that pair $(\mathcal{V}, \mathcal{B})$ does indeed form a KS set (there are an odd number of orthogonal bases, and an even number of vectors in each orthogonal basis). Note that the vectors $v^{\{i,j\}}$ generated by this construction may not be pairwise linearly independent. 

We say that an orthogonal representation is \emph{faithful} if non-adjacent vertices are mapped to by non-orthogonal vectors. Whenever we may produce a faithful orthogonal representation of a graph $G$ via a KS set $\mathcal{K}$, we shall say that $\mathcal{K}$ \emph{induces} a faithful orthogonal representation of $G$.

Using the procedure that we have described with the matrix $D^{(3)}$, we obtain the KS set given in Table~\ref{tab:raysJ72}. This induces a faithful orthogonal representation of the Johnson graph $J(7,2)$. By construction the co-ordinates of each $v^{\{i,j\}}$ are of the form $\zeta_3^{m}$ for some $m \in \{0,1,2\}$. In order to simplify the presentation, we replace each $\zeta_3^{m}$ by its exponent $m$ in the vector so that for example $v^{\{1,4\}}$ is denoted by $(1, 0, 2, 2, 1, 0)$ but is really $(\zeta_3, 1, \zeta_3^2, \zeta_3^2, \zeta_3, 1)$. The vectors are written in this form in Table~\ref{tab:raysJ72}. 

We emphasize that the graph $J(7,2)$ is simply the graph obtained from the $J(7,2)$-configuration (see Fig.~\ref{fig:J72}) by viewing each point of the $J(7,2)$-configuration as a vertex, with a pair of vertices being adjacent whenever the corresponding points are colinear in the $J(7,2)$-configuration.

We also give the KS set obtained from the Jungnickel-Lison\v ek construction defining an orthogonal representation of $J(11,2)$ in Table~\ref{tab:raysJ112}. The orthogonal representation in this case is not faithful. For example, the vectors $v^{\{1,2\}}$ and $v^{\{3,4\}}$ are orthogonal but are not in a common maximum clique. The components of each vector are all of the form $\zeta_5^{m}$, where $\zeta_5$ is a primitive $5$th root of unity and $m \in \{0, 1, 2, 3, 4\}$. As before, we replace the power of root of unity with its exponent (i.e., $\zeta_5^m$ with $m$) in the presentation of the vectors for simplicity.

To summarize, whenever there is a generalized Hadamard matrix $GH(g,\lambda)$ over $\mathbb{Z}_g$ for some positive integer $\lambda$ satisfying $g\lambda = N-1$, we obtain a PQS for the colored odd-pointed star game played on the $J(N,2)$-configuration. Therefore, the Jungnickel-Lison\v ek construction ensures that there is a perfect quantum strategy for infinitely many $N$ (those of the form $2^kp^m+1$ with $m \geq k \geq 1$, $m,k$ integers). There are also generalized Hadamard matrices over $\mathbb{Z}_g$ so that $g\lambda = N-1$ for different values of $N$ than those we have mentioned; our intention here was to illustrate that there is an infinite class, for which our choice is enough. There are several other constructions of infinite families of generalized Hadamard matrices given in \cite[Table 5.10]{Colbourn:1995}. The reader should take care to note that these constructions are given over $EA(g)$ for general prime powers $g$, so one must restrict $g$ to be prime. Additionally, \cite[Theorems 5.11 and 5.12]{Colbourn:1995} provide recursive constructions for generalized Hadamard matrices. For smaller cases, generalized Hadamard matrices over $\mathbb{Z}_g$ are known to exist for each order $k := g\lambda =1, \dots, 100$, with the exception of $2, 4, 8, 32, 40, 42, 60, 64, 66, 70, 78, 84, 88$ \cite[Table 5.13]{Colbourn:1995} (note that, in some of these cases, it is still open whether such a generalized Hadamard matrix exists). Therefore, the smallest Johnson graph for which we may not produce a labeling via this method is $J(9,2)$. We refer the reader to the Sec.~5.2 of 
Ref.~\cite{Colbourn:1995} for a more complete list of the possible choices of $N$. 

Finally, let us remark that Lison\v ek's construction of a KS set does not \emph{need} a generalized Hadamard matrix, but only an S-Hadamard matrix. Note that here we refer to the ``Lison\v ek construction'' as opposed to the ``Jungnickel-Lison\v ek construction'' since we are not restricting the generalized Hadamard matrix to the ones built from Jungnickel's construction. 
It may be possible that for certain choices of $N$ suitable S-Hadamard matrices exist even when generalized Hadamard matrices do not.


\begin{table}[t]
    \centering
    {\tiny
    \begin{tabular}{ll}
        $v^{1, 2} = (0, 0, 0, 0, 0, 0, 0, 0, 0, 0)$ &
        $v^{1, 3} = (4, 0, 1, 2, 3, 3, 4, 0, 1, 2)$ \\
        $v^{1, 4} = (1, 3, 0, 2, 4, 2, 4, 1, 3, 0)$ & 
        $v^{1, 5} = (1, 4, 2, 0, 3, 2, 0, 3, 1, 4)$\\
        $v^{1, 6} = (4, 3, 2, 1, 0, 3, 2, 1, 0, 4)$ & 
        $v^{1, 7} = (0, 4, 1, 1, 4, 0, 2, 3, 3, 2)$\\
        $v^{1, 8} = (1, 1, 4, 0, 4, 3, 3, 2, 0, 2)$ & 
        $v^{1, 9} = (4, 0, 4, 1, 1, 2, 0, 2, 3, 3)$\\
        $v^{1, 10} = (4, 1, 1, 4, 0, 2, 3, 3, 2, 0)$ & 
        $v^{1, 11} = (1, 4, 0, 4, 1, 3, 2, 0, 2, 3)$\\
        $v^{2, 3} = (3, 0, 2, 4, 1, 1, 3, 0, 2, 4)$ & 
        $v^{2, 4} = (2, 1, 0, 4, 3, 4, 3, 2, 1, 0)$\\
        $v^{2, 5} = (2, 3, 4, 0, 1, 4, 0, 1, 2, 3)$ &
        $v^{2, 6} = (3, 1, 4, 2, 0, 1, 4, 2, 0, 3)$\\
        $v^{2, 7} = (0, 3, 2, 2, 3, 0, 4, 1, 1, 4)$ &
        $v^{2, 8} = (2, 2, 3, 0, 3, 1, 1, 4, 0, 4)$\\
        $v^{2, 9} = (3, 0, 3, 2, 2, 4, 0, 4, 1, 1)$ &
        $v^{2, 10} = (3, 2, 2, 3, 0, 4, 1, 1, 4, 0)$\\
        $v^{2, 11} = (2, 3, 0, 3, 2, 1, 4, 0, 4, 1)$ &
        $v^{3, 4} = (0, 3, 1, 4, 2, 0, 3, 1, 4, 2)$\\
        $v^{3, 5} = (0, 4, 3, 2, 1, 0, 4, 3, 2, 1)$ &
        $v^{3, 6} = (3, 3, 3, 3, 3, 1, 1, 1, 1, 1)$\\
        $v^{3, 7} = (4, 4, 2, 3, 2, 3, 1, 3, 4, 4)$ &
        $v^{3, 8} = (0, 1, 0, 2, 2, 1, 2, 2, 1, 4)$\\
        $v^{3, 9} = (3, 0, 0, 3, 4, 0, 4, 2, 4, 0)$ &
        $v^{3, 10} = (3, 1, 2, 1, 3, 0, 2, 3, 3, 2)$\\
        $v^{3, 11} = (0, 4, 1, 1, 4, 1, 1, 0, 3, 0)$ &
        $v^{4, 5} = (2, 2, 2, 2, 2, 4, 4, 4, 4, 4)$\\
        $v^{4, 6} = (0, 1, 2, 3, 4, 0, 1, 2, 3, 4)$ &
        $v^{4, 7} = (1, 2, 1, 3, 3, 2, 1, 4, 1, 2)$\\
        $v^{4, 8} = (2, 4, 4, 2, 3, 0, 2, 3, 3, 2)$ &
        $v^{4, 9} = (0, 3, 4, 3, 0, 4, 4, 3, 1, 3)$\\
        $v^{4, 10} = (0, 4, 1, 1, 4, 4, 2, 4, 0, 0)$ &
        $v^{4, 11} = (2, 2, 0, 1, 0, 0, 1, 1, 0, 3)$\\
        $v^{5, 6} = (0, 2, 4, 1, 3, 0, 2, 4, 1, 3)$ &
        $v^{5, 7} = (1, 3, 3, 1, 2, 2, 2, 1, 4, 1)$\\
        $v^{5, 8} = (2, 0, 1, 0, 2, 0, 3, 0, 1, 1)$ &
        $v^{5, 9} = (0, 4, 1, 1, 4, 4, 0, 0, 4, 2)$\\
        $v^{5, 10} = (0, 0, 3, 4, 3, 4, 3, 1, 3, 4)$ &
        $v^{5, 11} = (2, 3, 2, 4, 4, 0, 2, 3, 3, 2)$\\
        $v^{6, 7} = (4, 2, 3, 2, 4, 3, 4, 4, 3, 1)$ &
        $v^{6, 8} = (0, 4, 1, 1, 4, 1, 0, 3, 0, 1)$\\
        $v^{6, 9} = (3, 3, 1, 2, 1, 0, 2, 3, 3, 2)$ &
        $v^{6, 10} = (3, 4, 3, 0, 0, 0, 0, 4, 2, 4)$\\
        $v^{6, 11} = (0, 2, 2, 0, 1, 1, 4, 1, 2, 2)$ &
        $v^{7, 8} = (1, 0, 0, 1, 3, 3, 0, 0, 3, 4)$\\
        $v^{7, 9} = (4, 4, 0, 2, 0, 2, 2, 0, 1, 0)$ &
        $v^{7, 10} = (4, 0, 2, 0, 4, 2, 0, 1, 0, 2)$\\
        $v^{7, 11} = (1, 3, 1, 0, 0, 3, 4, 3, 0, 0)$ &
        $v^{8, 9} = (0, 1, 3, 1, 0, 0, 3, 4, 3, 0)$\\
        $v^{8, 10} = (0, 2, 0, 4, 4, 0, 1, 0, 2, 2)$ &
        $v^{8, 11} = (2, 0, 4, 4, 0, 1, 0, 2, 2, 0)$\\
        $v^{9, 10} = (3, 1, 0, 0, 1, 4, 3, 0, 0, 3)$ &
        $v^{9, 11} = (0, 4, 4, 0, 2, 0, 2, 2, 0, 1)$\\
        $v^{10, 11} = (0, 0, 1, 3, 1, 0, 0, 3, 4, 3)$ & \
    \end{tabular}
    }
    \caption{The vectors of the KS set produced by the Jungnickel-Lison\v ek construction for the case $q=5$ (and non-square element of $GF(5),\ n=2$ in the construction of the generalized Hadamard matrix). We have replaced each coordinate by its exponent (so for $\zeta_5^m$ we simply write the exponent $m$).
    These vectors form an orthogonal representation of $J(11,2)$. The orthogonal representation is not faithful.}
    \label{tab:raysJ112}
\end{table}


\begin{figure*}[htbp]
    \flushleft
    \vspace{-1cm}
    \hspace{0.6cm}
    \begin{minipage}[b]{0.35\textwidth}
        \centering
        \begin{subfigure}[b]{\textwidth}
            \centering
            \captionsetup{margin={3cm,0cm}}
            \includegraphics[scale=0.6]{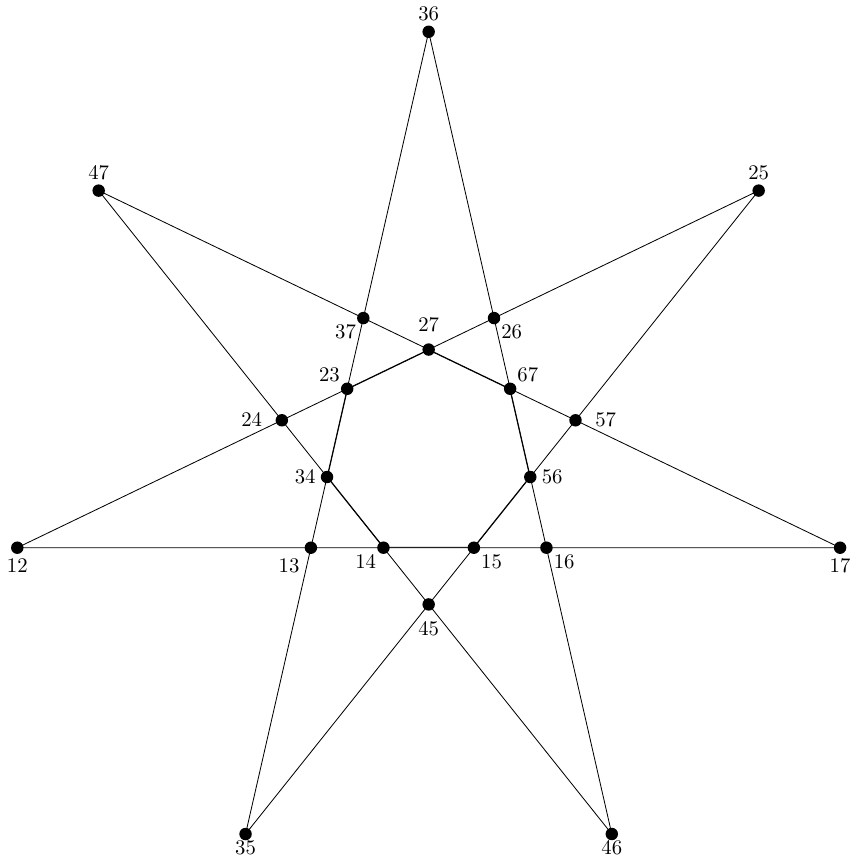}
            \caption{}
            \label{fig:J72}
        \end{subfigure}
    \end{minipage}

    \hspace{-0.6cm}

    \vspace{2.5cm}
    \begin{minipage}[t]{0.35\textwidth}
        \begin{subfigure}[b]{\textwidth}
            \centering
            \captionsetup{margin={4.2cm,0cm}}
            \includegraphics[scale=0.6]{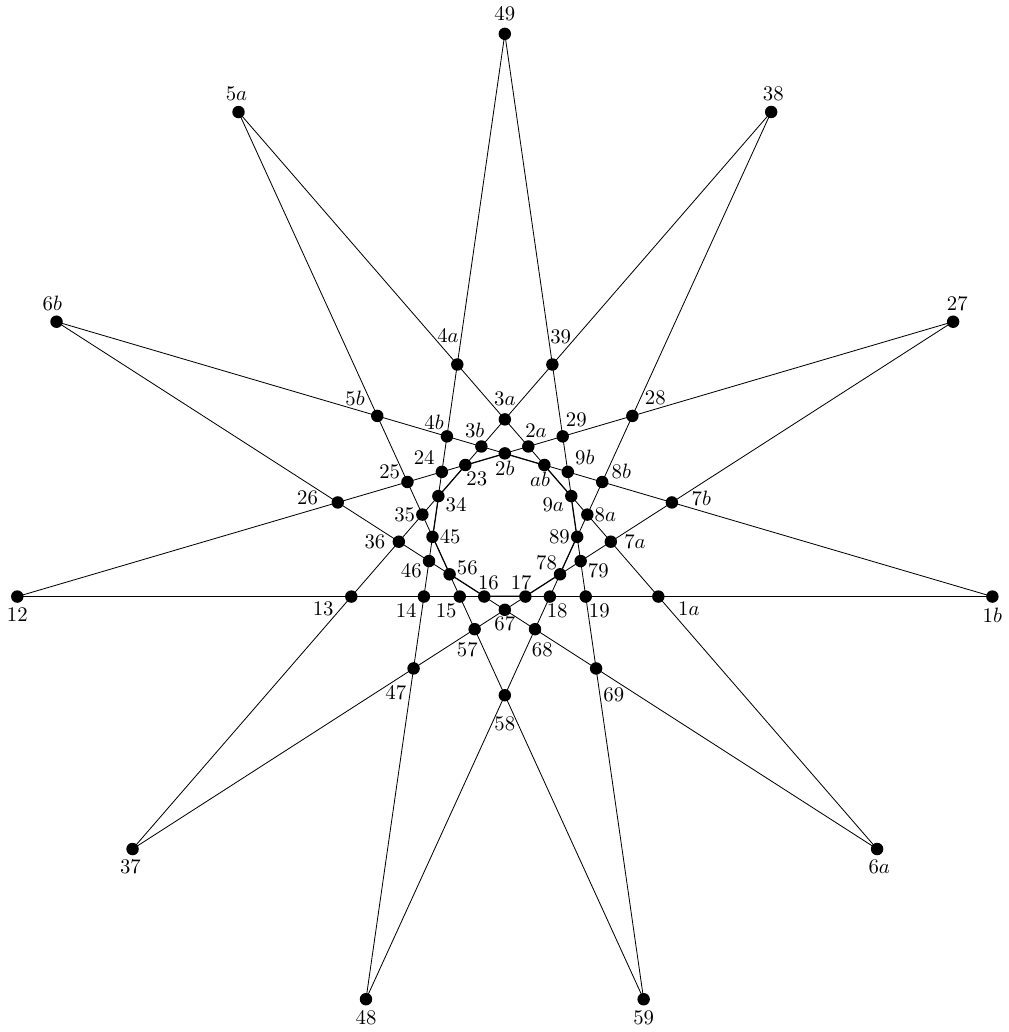}
            \caption{}
            \label{fig:fig3b.pdf}
        \end{subfigure}
    \end{minipage}

    \flushright
    \vspace{-17cm}
    \begin{minipage}[b]{0.45\textwidth}
        \centering
        \begin{subfigure}[b]{\textwidth}
            \centering
            \captionsetup{margin={1.4cm,0cm}}
            \includegraphics[scale=0.65]{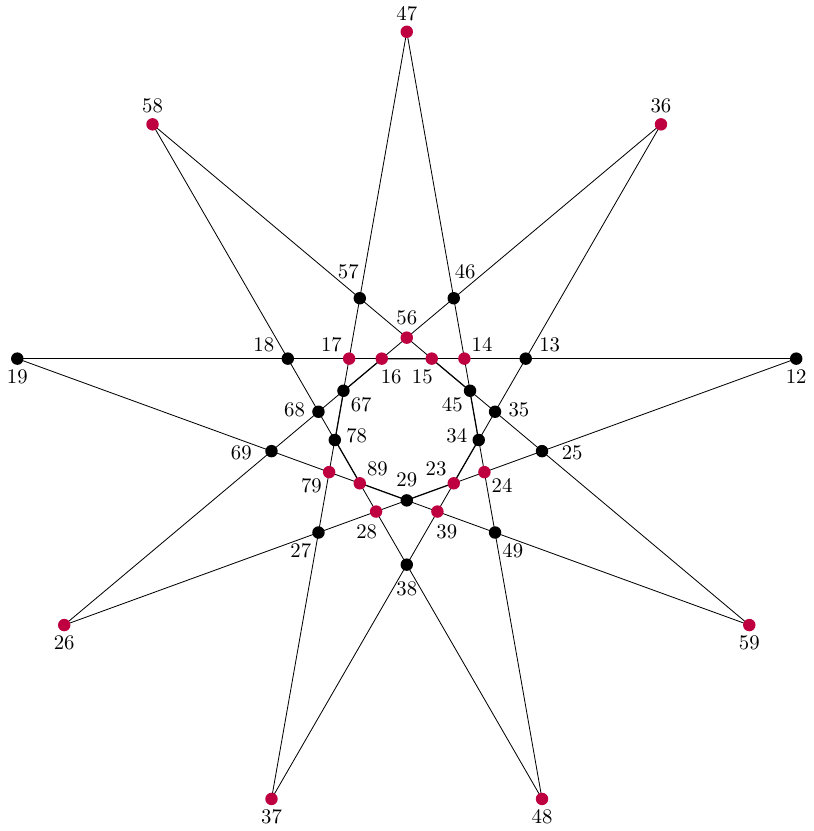}
            \caption{}
            \label{fig:J112}
        \end{subfigure}
        \vfill
    \end{minipage}

    \vspace{6.7cm}
    \caption{The $J(N,2)$-configurations for $N=7,11$ are illustrated in (a) and (b), respectively. Two vertices of $J(N ,2)$ are adjacent if and only if the corresponding points appear on a common line. For simplicity, vertex $\{i,j\}$ is labeled by $ij$. The $J(9,2)$-configuration is shown in (c). The $J(9,2)$-configuration is the simplest case where the Jungnickel-Lison\v ek construction cannot be used to produce a KS set. However, the orthogonal representation in Table~\ref{tab:raysJ92} does. This orthogonal representation is not faithful. It is obtained 
    from the 18-vector KS set in $d=4$ \cite{Cabello:1996PLA} by replacing each vector $a$ with the two vectors $(a,0)$ (purple dots) and $(0,a)$ (black dots). }
    \label{fig:JohsonGraphs}
\end{figure*}

\begin{figure}
    \centering
    \begin{subfigure}[b]{0.22\textwidth}
        \centering
        \includegraphics[scale=1.2]{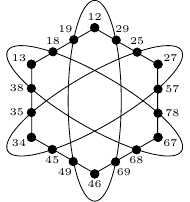}
        \caption{}
        \label{fig:CEG-18_A}
    \end{subfigure}
    \hfill
    \begin{subfigure}[b]{0.22\textwidth}
        \centering
        \includegraphics[scale=1.2]{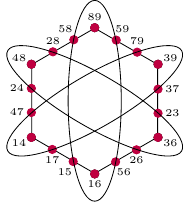}
        \caption{}
        \label{fig:CEG-18_B}
    \end{subfigure}
    \caption{The two copies of the 18-vector KS set \cite{Cabello:1996PLA} that appear in Fig.~\ref{fig:JohsonGraphs} (c). Two vertices are adjacent if and only if the corresponding points appear on a common line (which here can be a segment or an ellipse).}
    \label{fig:CEG-18_pair}
\end{figure}


\subsection{A novel construction for dimensions not covered by the Jungnickel-Lison\v ek construction}


The $N=9$ case is the smallest example of where there is no suitable generalized Hadamard matrix to produce a KS set for $J(N,2)$. Here, we show that, in fact, there is a KS set for $J(9,2)$. Additionally, we provide a novel construction that allows one to produce a KS set for $J(N^k,2)$ for each non-negative integer $k$ whenever one has a KS set for $J(N,2)$ with $N$ a prime power. This allows us to extend the Jungnickel-Lison\v ek construction in infinitely many dimensions. As an example, we are able to produce a KS set for $J(121,2)$. The direct constructions illustrated in \cite[Table 5.10]{Colbourn:1995} do not produce a generalized Hadamard matrix in dimension $120$ for the Jungnickel-Lison\v ek construction. In fact, the recursive construction we present produces a KS set for $J(N',2)$ for each $N' \equiv N \mod N(N-1)$ whenever $N'$ is sufficiently large.

Given a graph $G = (V,E)$, the \emph{line graph} of $G$, denoted $L(G) = (V',E')$ is the graph whose vertices are the edges of $G$ (i.e $V' = E$) and for which two vertices $w,x \in V'$ are adjacent if and only if $w$ and $x$ are incident to a common vertex of $G.$ The Johnson graph $J(N,2)$ is the line graph of the complete graph $K_N.$ 

A subgraph $H$ of $G$ is a \emph{spanning} graph of $G$ if $V(H) = V(G)$. A $k$-regular (each vertex has degree $k$) spanning graph of $G$ is called a \emph{$k$-factor} of $G$, and a \emph{$k$-factorization} of $G$ is a set of $k$-factors of $G$, $\{H_1, \dots, H_{\ell}\}$ whose edge sets form a partition of the edge set of $G$.

\begin{proposition} \label{prop:factor}
    Let $G$ be a graph with a $k$-factorization $H_1, \dots, H_{\ell}$ such that $L(H_i)$ has an orthogonal representation in $\mathbb{C}^k$ for each $i=1,\dots,\ell$. Then, there is an orthogonal representation for $L(G)$ in $\mathbb{C}^{k\ell}$.
\end{proposition}

\begin{proof}
    Let $|E(G)| = n'$, so $n := n'/k = |E(H_i)| $ for each $i=1,\dots,\ell$.  Further, let
    \begin{equation}
        \{v_j^{i} : j =1, \dots, n\}
    \end{equation}
    be the $k$-dimensional orthogonal representation of $L(H_i)$ for each $i=1,\dots,\ell$. Each of these vectors is associated to a vertex of $L(H_i)$, and thus to an edge of $H_i$, and thus an edge of $G$. Now define vectors 
    \begin{equation}
        \overline{v_j^{i}} := (\underbrace{0, \dots, 0}_{i-1}, v_j^{i}, \underbrace{0, \dots, 0}_{\ell - i}),
    \end{equation}
    where $0$ denotes the $k$-dimensional zero vector and $v_j^{i}$ is the $i$th block of size $k$. As each of these vectors are associated to an edge of $G$, they are associated to vertices of $L(G)$. It is straightforward to check that a pair of such vectors is orthogonal whenever the corresponding vertices of $L(G)$ are adjacent. Additionally, by construction, the dimension of each each of these vectors is $k\ell$. Therefore, we see that the vectors $\overline{v_j^{i}}$ do indeed form a $k\ell$-dimensional orthogonal representation of $G$ as required.
\end{proof}

Consider the $4$-factorization of $K_9$ into two isomorphic copies of the \emph{Paley graph} of order $9$ (henceforth $P(9)$) as shown in Fig.~\ref{fig:K9-factorization}. Paley graphs are important objects in graph theory - among other important properties, they are self-complementary and strongly regular. 
The Paley graph of order $9$ encodes the commutativity relationships of the Peres-Mermin ``magic'' square \cite{Peres:1990PLA,Mermin:1990PRLb}, and its line graph encodes all the $54$ ``essential'' orthogonalities of the $18$ vector KS set of \cite{Cabello:1996PLA} (the 18 vector KS set actually has $63$ orthogonalities, but one may ignore $9$ of them and still find that no assignment $f$ may be produced). Indeed, the $18$ vector KS set provides a $4$-dimensional orthogonal representation of $P(9)$.
Applying Proposition \ref{prop:factor}, we immediately obtain the $8$-dimensional orthogonal representation given in Table \ref{tab:raysJ92} which is a KS set for $J(9,2)$. Fig.~\ref{fig:CEG-18_pair} illustrates explicitly the two copies [subfigures~(a) and (b)] of the 18-vector KS set \cite{Cabello:1996PLA} that appear in the $J(9,2)$-configuration. Here, the 18-vector KS set appears in the same format as in \cite{Cabello:2008PRL} and \cite{Budroni:2022RMP}, where the $9$ orthogonal bases are given by the $6$ straight lines and $3$ ellipses. The two copies are further emphasized in our depiction of the $J(9,2)$-configuration of Fig.~\ref{fig:JohsonGraphs}~(c) by displaying the vertices corresponding to each copy in the same colors (black and purple) as in Fig.~\ref{fig:CEG-18_pair}. 


\begin{figure*}
    \centering
    \begin{subfigure}[b]{0.39\textwidth}
        \centering
        \includegraphics[scale=0.4]{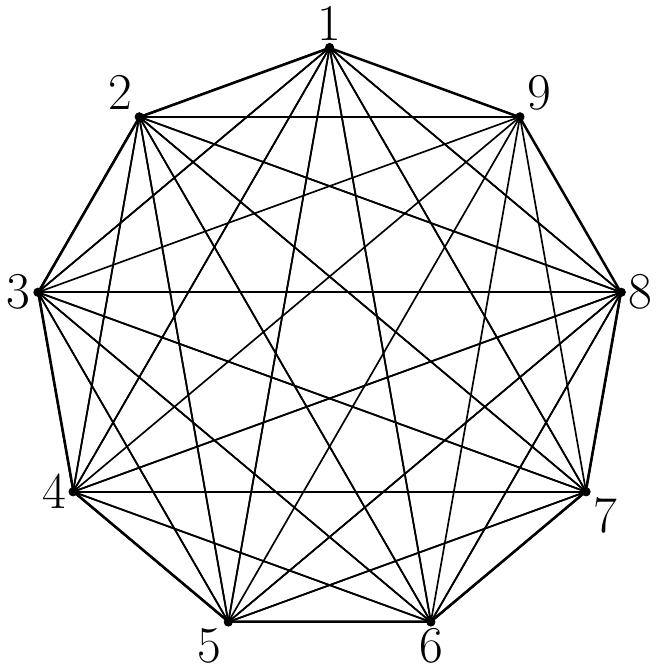}
        \caption{}
        \label{fig:K9}
    \end{subfigure}
    \hfill
    \begin{subfigure}[b]{0.59\textwidth}
        \centering
        \includegraphics[scale=0.8]{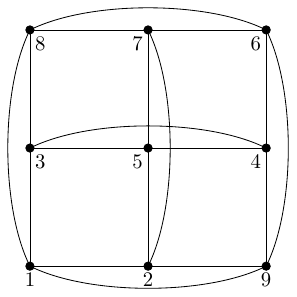}
        \includegraphics[scale=0.8]{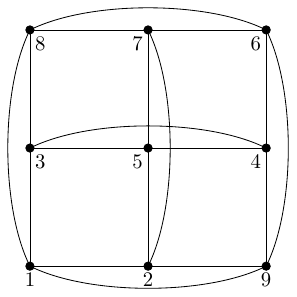}
        \caption{}
        \label{fig:rook_graph1}
    \end{subfigure}
    \caption{A $4$-factorization of (a) the complete graph $K_9$ into (b) two isomorphic copies of the Paley graph of order $9$, $P(9)$. Each edge of $K_9$ appears in exactly one copy of $P(9)$.}
    \label{fig:K9-factorization}
\end{figure*}


We now introduce a general construction for building KS sets for Johnson configurations from smaller ones. This construction is based on \emph{balanced incomplete block designs} (BIBDs). 

\begin{definition}
    Given positive integers $v,b,r,k,\lambda$, a \emph{balanced incomplete block design} is a pair $(S,T)$ where:
    \begin{enumerate}
        \item $S$ is a set of size $v$,
        \item $T$ is a collection of $b$ subsets of size $k$ called \emph{blocks},
        \item each element of $S$ is in exactly $r$ blocks,
        \item each pair of elements of $S$ is in exactly $\lambda$ blocks.
    \end{enumerate}
    If there is a partition $(t_1, \dots, t_m)$ of the blocks of $T$ so that each set of blocks $t_i$ is a partition of $S$, then we say that the BIBD is \emph{resolvable}, and write RBIBD for short. Finally, we call the partition $(t_1, \dots, t_m)$ a \emph{resolution}.
\end{definition}

Note that the parameters $v,k,\lambda$ determine the others, so in case we want to specify the parameters, we write $(v,k,\lambda)$-RBIBD. 

The key point is that an $(N,k,1)$-RBIBD yields a $(k-1)$-factorization of $K_N$, via the following recipe.
\begin{itemize}
    \item Let $S$ be the set of vertices of $K_N$.
    \item Consider some $t_i$ in the resolution. It is a partition of the vertices into sets of size $k$. Associate to each of these the complete graph $K_k$, and associate to $t_i$ the (vertex-disjoint) union of all of these complete graphs. Then $t_i$ corresponds to a $(k-1)$-regular spanning subgraph of $K_N$. Call this graph $H_i$.
    \item Since $\lambda = 1$, the graphs $H_1, \dots, H_m$ are edge-disjoint and the union of their edge sets is the edge set of $K_N$. Therefore $H_1, \dots, H_m$ is a $(k-1)$-factorization of $K_N$.
\end{itemize}

Notice further that since each $H_i$ is a disjoint union consisting of isomorphic copies of the graph $K_k$, a $(k-1)$-dimensional orthogonal representation of $L(K_k)$ provides a $(k-1)$-dimensional orthgonal representation of $L(H_i)$ -- one can simply use the same vectors for each copy of $K_k$. By applying Proposition \ref{prop:factor}, we obtain the following result.

\begin{theorem}
    If there exists a $(N,k,1)$-RBIBD for some odd $k$, and a $(k-1)$-dimensional orthogonal representation of $J(k,2)$, then there is an $(N-1)$-dimensional orthogonal representation of $J(N,2)$.
\end{theorem}
Such RBIBDs exist for each choice of $k$ that is a prime power, since the \emph{affine geometry} $AG(2, k)$ is a $(k^2,k,1)$-RBIBD (see for example \cite[Theorem 2.16]{Colbourn:1995}). As a consequence, we immediately obtain a $120$-dimensional orthogonal representation of $J(121)$ which notably does not seem to be attainable via the generalized Hadamard constructions that we have seen. 

In fact, one can generalize this to all powers of the odd prime power $k$: there is a $(k^m, k, 1)$-RBIBD for each positive integer $m$ \cite[Eq.5.6.a]{Beth:1999}.

Combining this with the previous construction, we obtain the following:

\begin{theorem}
    Let $n$ be an odd prime power such that there is a generalized Hadamard matrix over $\mathbb{Z}_p$ for some $p | (n-1)$. Then, there is a $(n^m~-~1)$-dimensional orthogonal representation of the Johnson graph $J(n^m)$ for each positive integer $m$, and consequently a KS set associated to each of these Johnson graphs.
\end{theorem}

In fact, for any choice of $k$, and $N \equiv k \mod k(k-1)$ sufficiently large, there exists an $(N,k,1)$-RBIBD \cite{Chaudhuri:1973}. Therefore, we obtain orthogonal representations for $J(N,2)$ in dimension $N-1$ for infinitely many dimensions for each choice of $k$ for which we have an orthogonal representation in dimension $k-1$. To emphasize, using the $k=7$ representation we have given in Table \ref{tab:raysJ72}, one immediately obtains KS sets for $J(N,2)$ for each $N \equiv 7 \mod 42$ with $N$ sufficiently large.

The best known lower bound for ``sufficiently large'' is currently $\exp(\exp(k^{12k^2}))$ \cite{Yanxun:2000}. This bound can likely be improved -- indeed in the $k=7$ case the largest possible $N$ for which there may be no $(N,k,1)$-RBIBD is $294427$ \cite[Table 7.35]{Colbourn:1995}. We also remark (although no suitable orthogonal representation for $J(k,2)$ exists in these cases) that in the case of $k=5$ there is a $(N, 5, 1)$-RBIBD for each $N > 645$ satisfying the congruence condition and in the case of $k=3,4$ the congruence condition is not just necessary but also sufficient. Such explicit bounds do not seem to have been obtained for other $k$ (see \cite[Remark 7.36]{Colbourn:1995}).


\begin{table}[t]
    \centering
    {\tiny
    \begin{tabular}{ll}
        $v^{1,2} = (0, 0, 0, 0, 0, 0, 0, 1)$ &
        $v^{1,3} = (0, 0, 0, 0, 0, 1, -1, 0)$ \\
        $v^{1,4} = (0, 0, 0, 1, 0, 0, 0, 0)$ &
        $v^{1,5} = (0, 1, -1, 0, 0, 0, 0, 0)$ \\
        $v^{1,6} = (0, 1, 1, 0, 0, 0, 0, 0)$ &
        $v^{1,7} = (1, 0, 0, 0, 0, 0, 0, 0)$ \\
        $v^{1,8} = (0, 0, 0, 0, 0, 1, 1, 0)$ &
        $v^{1,9} = (0, 0, 0, 0, 1, 0, 0, 0)$ \\
        $v^{2,3} = (1, 1, 1, 1, 0, 0, 0, 0)$ &
        $v^{2,4} = (1, 0, -1, 0, 0, 0, 0, 0)$ \\
        $v^{2,5} = (0, 0, 0, 0, 1, 0, 1, 0)$ &
        $v^{2,6} = (1, -1, 1, -1, 0, 0, 0, 0)$ \\
        $v^{2,7} = (0, 0, 0, 0, 1, 0, -1, 0)$ &
        $v^{2,8} = (0, 1, 0, -1, 0, 0, 0, 0)$\\
        $v^{2,9} = 0, 0, 0, 0, (0, 1, 0, 0)$ &
        $v^{3,4} = (0, 0, 0, 0, 1, 1, 1, -1)$ \\
        $v^{3,5} = (0, 0, 0, 0, -1, 1, 1, 1)$ &
        $v^{3,6} = (1, 1, -1, -1, 0, 0, 0, 0)$ \\
        $v^{3,7} = (0, 0, 1, -1, 0, 0, 0, 0)$ &
        $v^{3,8} = (0, 0, 0, 0, 1, 0, 0, 1)$ \\
        $v^{3,9} = (1, -1, 0, 0, 0, 0, 0, 0)$ &
        $v^{4,5} = (0, 0, 0, 0, 1, 1, -1, 1)$ \\
        $v^{4,6} = (0, 0, 0, 0, 1, -1, 0, 0)$ &
        $v^{4,7} = (0, 1, 0, 0, 0, 0, 0, 0)$ \\
        $v^{4,8} = (1, 0, 1, 0, 0, 0, 0, 0)$ &
        $v^{4,9} = (0, 0, 0, 0, 0, 0, 1, 1)$ \\
        $v^{5,6} = (1, 0, 0, 1, 0, 0, 0, 0)$ &
        $v^{5,7} = (0, 0, 0, 0, 0, 1, 0, -1)$ \\
        $v^{5,8} = (-1, 1, 1, 1, 0, 0, 0, 0)$ &
        $v^{5,9} = (1, 1, 1, -1, 0, 0, 0, 0)$ \\
        $v^{6,7} = (0, 0, 0, 0, 1, 1, 1, 1)$ &
        $v^{6,8} = (0, 0, 0, 0, 1, 1, -1, -1)$\\
        $v^{6,9} = (0, 0, 0, 0, 0, 0, 1, -1)$ &
        $v^{7,8} = (0, 0, 0, 0, 1, -1, 1, -1)$ \\
        $v^{7,9} = (0, 0, 1, 1, 0, 0, 0, 0)$ &
        $v^{8,9} = (1, 1, -1, 1, 0, 0, 0, 0)$ \\
    \end{tabular}
    }
    \caption{The vectors of the KS set produced from the 18 vector KS set of \cite{Cabello:1996PLA} by replacing each vector $a$ with the two vectors $(a,0)$ and $(0,a)$. 
    As shown in Fig.~\ref{fig:JohsonGraphs}~(c), these vectors form an orthogonal representation of $J(9,2)$. The orthogonal representation is not faithful.} 
    \label{tab:raysJ92}
\end{table}


Aside from the $J(9,2)$ case, the smallest case for which the there is no suitable generalized Hadamard matrix is $J(33,2)$ (in dimension $32$). 
Future research is needed to check if a suitable $S$-Hadamard matrix exists in this case, or if one may leverage a $k$-factorization of $K_{33}$ as we have done in this section.


\section{Perfect quantum strategies with minimal input cardinality} \label{sec.Optimality}


Let us now go back to the bipartite (colored odd-pointed star) games. We will finally justify our choice of input sets for Alice and Bob. In this section, $b$ will be used to denote an orthogonal basis of a KS set.

Let us assume that we have an orthogonal representation of $J(N,2)$ in dimension $N-1$. Denote the vectors by $\mathcal{V} = \{v^{i,j} : 1 \leq i < j \leq N\}$ with each $v^{i,j}$ associated to vertex $\{i,j\}$ of $J(N,2)$. Then, it is straightforward to see that, for each $i=1,\dots,N$, the set $b_i = \{v^{i,j} : j =1, \dots, N, j \neq i\}$ forms an orthogonal basis. Let $\mathcal{B} := \{b_1, \dots, b_N\}$. From Sec.~\ref{sec.InfiniteClass}, we know that $(\mathcal{V}, \mathcal{B})$ forms a KS set. The natural generalization of the PQS outlined in Sec,~\ref{sec.N=7} yields a perfect quantum strategy for the colored odd-pointed star game played on the $J(N,2)$-configuration. 

Consider sets $S_A, S_B~\subseteq~\mathcal{B}$. A \emph{classical assignment} consists of a pair  
\begin{align*}
    &f_A:\mathcal{V} \times S_A \to\{0, 1\} \\
    &f_B: \mathcal{V} \times S_B \to\{0,1\}
\end{align*}
satisfying
\begin{enumerate}
    \item for $b_i \in S_A, b_k \in S_B$, and vectors $v^{i,j} \in b_i, v^{k,\ell} \in b_k$, if $|\{i,j\} \cap \{k,\ell\}| = 1$, then we cannot have $f_A(v^{i,j}, b_i) = f_B(v^{k,\ell}, b_k) = 1,$
    \item
    \begin{enumerate}[(a)]
        \item for each orthogonal basis $b_i \in S_A$,\ $\sum_{j \neq i} f_A(v^{i,j}, b_i) = 1 $,
        \item for each orthogonal basis $b_k~\in~S_B,\ \sum_{\ell \neq k} f_B(v^{k, \ell}, b_k) = 1$.
\end{enumerate} 
\end{enumerate}
Note that we tacitly assume that $f_A(v,b) = 0$ for each $b \in S_A$ if $v \notin b$, and $f_B(v,b) = 0$ for each $b \in S_B$ if $v \notin b$. 

Given a classical assignment $f_A, f_B$, we say that an orthogonal basis $b_i$ is \emph{labeled} for $S_A$ if we have specified the values $f_A(v^{i,j},b_i)$ for each of the vectors $v^{i,j} \in b_i$ ($i=1, \dots, N, i \neq j$). In order to do this, it is sufficient to specifiy the vector $v^* \in b_i$ with $f_A(v^*,b_i) = 1$. Similarly, we say that $b_i$ is \emph{unlabeled} for $S_A$ if we have not yet designated the values $f_A(v^{i,j},b_i)$. Of course, the same holds replacing $A$ with $B$. 

\begin{definition}
\label{def:bks}
     The pair $(S_A, S_B)$ is a \emph{bipartite KS set} (B-KS) if no classical assignment exists for $(S_A, S_B).$
\end{definition} 

 It is clear that $(\mathcal{B}, \mathcal{B})$ is always a B-KS since $(\mathcal{V}, \mathcal{B})$ is a KS set,. Additionally, if $S_B \subseteq S_B'$ for some $S_B' \subseteq \mathcal{B}$, then $(S_A, S_B')$ is also a B-KS.
A B-KS set $(S_A, S_B)$ is an \emph{optimal bipartite KS set} if $|S_A||S_B|$ is minimum over all choices of $S_A, S_B \subseteq \mathcal{B}$.

The relationship between B-KS sets and the colored odd-pointed star game is as follows:
\begin{enumerate}
    \item B-KS sets correspond exactly to the input sets for Alice and Bob in the colored odd-pointed star game for which there is no perfect classical strategy,
    \item optimal B-KS sets correspond exactly to such input sets of minimum input cardinality.
\end{enumerate}

Throughout the proofs, we take advantage of the symmetries of the graph $J(N,2)$. Namely, any permutation $\sigma \in \mathfrak{S}_N$ of $\{1,2,\dots,N\}$ induces an automorphism of $J(N,2)$ by mapping each vertex $ij$ to $\sigma(i)\sigma(j)$. It therefore also maps the maximum clique $C_i := \{ij: j=1,\dots,N, j \neq i\}$ to the maximum clique $C_{\sigma(i)} := \{\sigma(i)j: j=1,\dots,N, j \neq \sigma(i)\}$. This induces an automorphism of the KS set mapping each $v^{i,j}$ to $v^{\sigma(i), \sigma(j)}$, and consequently each $b_i$ to $b_{\sigma(i)}$.

\begin{theorem} \label{theo:optimal-labeling}
    The optimal B-KS sets are exactly the pairs $(S_A, S_B)$ satisfying:
    \begin{align}
        S_A &= \{b_m : m \notin \{i,j\}\} \label{eq:S_A} \\
        S_B &= \{b_n : n \notin \{k,\ell\}\} \label{eq:S_B}
    \end{align}
    for any choice of distinct $i,j,k,\ell \in \{1,2,\dots,N\}.$ 
\end{theorem} 

\begin{proof}
We first show that $\min(|S_A|, |S_B|)\geq N-2$. Let us proceed by contradiction, assuming without loss of generality that $|S_A| \leq N - 3$, and showing that a classical assignment exists in such a case. We may assume that $S_B = \mathcal{B}$.
Additionally, by exploiting the symmetries of $J(N,2)$, we may assume without loss of generality that $b_1, b_2, b_3 \in \mathcal{B} \setminus S_A$.  Then set 
\begin{equation}
    f_B(v^{1,2}, b_1), f_B(v^{1,2}, b_2), f_B(v^{1,3}, b_3) := 1.
\end{equation}
Since $N$ is odd, this leaves an even number of orthogonal bases to label for $S_B$. By exploiting the fact that the intersection of each pair of orthogonal bases is a single vector, we are able to greedily label the remaining orthogonal bases of $S_B$ -- we simply pick an unlabeled pair of orthogonal bases $b_i, b_j$ and set:
\begin{align}
    f_B(v^{i,j},b_i), f_B(v^{i,j},b_j) := 1.
\end{align}
We continue this process until all orthogonal bases of $S_B$ have been labeled. 
Then, we label the orthogonal bases of $S_A$ (which are each in $S_B$) exactly the same as we have labeled the bases of $S_B$, so:
\begin{align}
    f_A(v,b) := f_B(v,b)
\end{align}
for each $b \in S_A \subseteq S_B, v\in b.$
This constitutes a classical assignment, and so we see that such a pair $(S_A, S_B)$ cannot be a B-KS set. Therefore, $\min(|S_A|, |S_B|) \geq N-2$.

Next, we show that the choices of $(S_A, S_B)$ described in the theorem statement are indeed bipartite KS sets. Let $i,j,k,\ell \in \{1,2,\dots,N\}$ be distinct, and let $S_A, S_B$ be defined as in Eqs.~\eqref{eq:S_A}--\eqref{eq:S_B}. By symmetry, we may assume without loss of generality that $i~=~1,j~=~2, k~=~3, \ell=4$, so that $S_A = \{b_3, b_4, \dots, b_N\}$ and $S_B = \{b_1, b_2, b_4, \dots, b_N\}.$

We prove that $(S_A, S_B)$ is a B-KS set by showing that the existence of a classical assignment $(f_A, f_B)$ of $(S_A, S_B)$ implies the existence of a classical assignment $(f'_A, f'_B)$ of $(\mathcal{B}, \mathcal{B})$. This would contradict the fact that $(\mathcal{B}, \mathcal{B})$ is a bipartite KS set, thus that $(\mathcal{V}, \mathcal{B}$) is a KS set, and finally that $\mathcal{V}$ is an orthogonal representation of $J(N,2)$ in dimension $N-1$. 

Since $(f_A, f_B)$ is a classical assignment, we must have that for each $b \in S_A \cap S_B$, 
\begin{equation}
    f_A(v,b) = f_B(v',b) \iff v = v'
\end{equation}
otherwise Condition 1 of the classical assignment definition is violated. Consider extending $(f_A, f_B)$ to an assignment of $(\mathcal{B}, \mathcal{B})$ by labeling the missing bases of $S_A$ with the choices made in $f_B$ (i.e setting $f_A(v,b_1) := f_B(v,b_1)$ for each $v \in b_1$, $f_A(v',b_2) := f_B(v',b_2)$ for each $v' \in b_2$) and vice-versa. This either forms a classical assignment of $(\mathcal{B}, \mathcal{B})$, or Condition 1 is violated in the process. In the case that a violation occurs, it must be between $b_1$ and $b_2$ or $b_3$ and $b_4$. This is because every other pair of bases has one point labeled by $f_A$ and the other by $f_B$. Therefore, we may assume without loss of generality that
\begin{equation}
    f_B(v^{1,2}, b_1) = f_B(v^{2,m}, b_2) = 1
\end{equation}
for some $m \geq 3$ so that $|\{1,2\} \cap \{2,m\}| = 1$. However such a scenario is impossible since there is a violation of Condition 1 for any choice of labeling of $b_m$ for $f_A$: setting $f_A(v^{m,m'},b_m) = 1$ for any $m' \neq m$, we find that either $|\{m,m'\} \cap \{1,2\}| = 1$ or $|\{m,m'\} \cap \{2,m\}| = 1.$ Thus, the classical assignment $(f_A, f_B)$ of $(S_A, S_B)$ may indeed be extended to an assignment $(f'_A, f'_B)$ of $(\mathcal{B}, \mathcal{B})$.

As stated earlier, this would imply that $\mathcal{V}$ does not define an orthogonal representation of $J(N,2)$ in dimension $N-1$, which is a contradiction. Therefore, $(S_A, S_B)$ must indeed be a B-KS set, and since $\min(|S_A|, |S_B|) \geq |\mathcal{B}| - 2$, this is optimal. 

Any other optimal B-KS set $(S'_A, S'_B)$ must also have $|S'_A| = |S'_B| = |\mathcal{B}| - 2$. If it is not of the type described above, then the B-KS set must correspond to two (possibly equal) vectors in a common orthogonal basis, say $b_N$ without loss of generality (again each vector corresponds exactly to the pair of orthogonal bases missing from $S'_A$ or $S'_B$). In this case, there is a simple classical assignment: setting $f_A(r_{i,i+1}, b_i) = f_B(r_{i,i+1}, b_i) = 1$ for each valid choice of even $i$, and $f_A(r_{j,j-1}, b_j) = f_B(r_{j,j-1}, b_j) = 1$ for each valid choice of odd $j$. By ``valid choice'' we mean that one should simply ignore assigning $f_A$ for the $b_i$ or $b_j$ that are not in $S_A$, and do the same replacing $A$ with $B$.

Therefore, we see that the optimal B-KS sets are exactly those sets $(S_A, S_B)$ in the statement of the theorem.
\end{proof}

In the context of the colored odd-pointed star games, we have shown the following. For any input choice of $N-2$ lines for Alice and $N-2$ lines for Bob with no line missing for both Alice and Bob (i.e each of the possible $N$ lines from the $J(N,2)$-configuration is given to either Alice or Bob or both), there is no perfect classical strategy. Moreover, these are exactly the choices with minimal input size. We record this below.

\begin{corrolary} \label{cor.Optimal}
    Giving each of Alice and Bob $N-2$ lines so that none of the $N$ lines are missing for both Alice and Bob yields the minimum input cardinality $|A||B|$ for the colored odd-pointed star games. With this choice of input sets $A,B$ there is a PQS but no perfect classical strategy. Any choice of input sets not of this form either has larger input cardinality $|A||B|$ or a perfect classical strategy.
\end{corrolary}


\section{Bell inequalities} \label{sec.BellIneqs}


Although the main interest of PQSs is that they are perfect rather than merely better than classical, a legitimate question is how much better than classical they are. This is equivalent to asking which are the Bell inequalities associated to the colored odd-pointed star games. Related to this is also the question \cite{Gisin:2007IJQI} of whether 
Bell inequalities associated to PQSs are tight (i.e., facets of the local polytope) or not. Recall that, for the simplest PQSs (the GHZ-Mermin \cite{GHZ89,Mermin:1990PT} and magic square \cite{Cabello:2001PRLb} correlations), the corresponding Bell inequalities (\cite{Mermin:1990PRLa,Cabello:2001PRLb}, respectively) are tight (as proven in \cite{Sliwa2003PLA,Gisin:2007IJQI}, respectively). However, recently, it has been shown \cite{Liu:2023XXX} that there are PQSs whose associated Bell inequalities are not tight.

In this section we present the general form of the Bell inequalities associated to the colored odd-pointed star games and prove that neither of them is tight.


\subsection{Bell inequalities associated to the colored odd-pointed star games}


Here we introduce the linear combination of probabilities, i.e., Bell inequality, that is associated to the game described in Sec.~\ref{sec.TheGame}. We first write the Bell inequality associated to the colored odd-pointed star game with $N=7$, and then we generalize the analysis for an arbitrary $N$. It is convenient for this purpose to use the configuration view of the colored odd-pointed star game.

In order to write the Bell inequality associated to the $J(7,2)$-configuration, we are going to momentarily ignore the line's colors and exploit the symmetries that appear when considering that Alice and Bob could receive any of the seven lines of the $J(7,2)$-configuration as input.

As before, Alice and Bob have six possible outputs, corresponding to the points in their respective lines. For convenience, we order the points of each Line $i$ in lexicographical order: so $\{i,j\}$ comes before $ \{i,j'\} $ if and only if $j < j'$. For example, for Line $5$ of the $N=7$ case, this produces the order $15, 25, 35, 45, 56, 57$.
This allows us to represent each of the outputs of Alice and Bob in a manner that does not depend on the input line. 
 
More precisely, we set Alice's outputs to be $a=1, \dots, 6$, where choosing $a=n$ corresponds to choosing the $n^{th}$ point on the input line in the order we have just described. To emphasize, if Alice outputs $a=6$ for Line $5$ described above, this corresponds to Alice choosing point $57$ (and \emph{not} point $56$ as in the notation used in previous sections). We emphasize this further in a remark to avoid any confusion.

\begin{remark}
    In this section only, the choice of output $n \in \{1,2,\dots,6\}$ (and more generally $\{1,2,\dots,N-1\}$) for either Alice and Bob corresponds to choosing the $n^{th}$ point on their given input line. 
\end{remark}

With these preparations, we can define the coefficients that represent the winning conditions of the game as stated in Sec.~\ref{sec.TheGame}:
\begin{equation}\label{eq.cabxy}
c(a,b|x,y) = \left\{ \begin{array}{r@{\quad}c@{\quad}l} \delta(a=b),&\mathrm{for}~x=y \\
c'(a,b|x,y),&\mathrm{for}~x<y \\
c(b,a|y,x),&\mathrm{for}~x>y, 
\end{array} \right.
\end{equation}
where $\delta(z)=1 (0)$ when $z=$True (False) and 
\begin{equation}
\begin{split}
    c'(a,b|x,y) &= \delta(a \neq b)\\
    &+ \delta(a=b)\delta(x\leq b \leq y-1)\\
    &- \delta(a=b-1)\delta(x<b<y)\\
    &- \delta(a=y-1)\delta(b\neq x) \\
    &- \delta(a \neq y-1)\delta(b=x),
\end{split}
\end{equation}    
where $a,b=1,\dots,6$ and $x,y=1,\dots,7$. 

Once we have the coefficients $c(a,b|x,y)$ defined, we can include the lines' colors and write the Bell inequality associated to the $J(7,2)$-configuration
\begin{equation}\label{eq.heptagram}
\sum_{x = \text{red}} \sum_{y = \text{blue}}\sum_{a,b} c(a,b|x,y) p(a,b|x,y) \leq 24,
\end{equation}
where the summation over red (blue) denotes that Alice (Bob) is going to only be inquired for the red (blue) lines. In the colored odd-pointed star games, with $N=7$, these summations represent $x=1,2,3,4,5$ and $y=3,4,5,6,7$.

To generalize this Bell inequality for the colored odd-pointed star game we define the possible inputs as before to be any choice from $\{1,2,\dots,N\}$ (where Alice will be given input from $\{1,2,\dots,N-2\}$ and Bob will be given input from $\{3,4,\dots,N\}$), and the output set to be $\{1,2,\dots,N-1\}$.

Then, we can write the colored odd-pointed star game as the following Bell inequality
\begin{equation}\label{eq.Ngram}
\sum_{x = \text{red}} \sum_{y = \text{blue}}\sum_{a,b} c(a,b|x,y) p(a,b|x,y) \leq (N-2)^2-1,
\end{equation}
with the coefficients $c(a,b|x,y)$ defined in Eq.~\eqref{eq.cabxy} for $a,b=1,\dots, N-1$ and $x,y=1,\dots,N$.


\subsection{Proof that the colored odd-pointed star Bell inequalities are not tight}


Here we give an upper bound to the maximum dimension of the faces defined by Eq.~\eqref{eq.heptagram} and we generalize this method for the case of the colored odd-pointed star game Eq.~\eqref{eq.Ngram}. As a consequence of this upper bound, we conclude that these hyperplanes do not define facets of the classical polytope, i.e., they are not tight Bell inequalities.

For our argument it is enough to focus on the lines that are common to Alice and Bob. For the $J(7,2)$-configuration these are $x',y'=\{3,4,5\}$. It is helpful to arrange these combinations in a table, see Table~\ref{tab:game1}, and write the coefficients $c(a,b|x,y) = M_{x,y}$ explicitly as matrices
\begin{equation}\label{eq.M01}
M_{3,4} = \begin{pmatrix}
0 & 1 & 0 & 1 & 1 & 1\\
1 & 0 & 0 & 1 & 1 & 1\\
0 & 0 & 1 & 0 & 0 & 0\\
1 & 1 & 0 & 0 & 1 & 1\\
1 & 1 & 0 & 1 & 0 & 1\\
1 & 1 & 0 & 1 & 1 & 0\\
\end{pmatrix},
\end{equation}
\begin{equation}\label{eq.M02}
M_{3,5} = \begin{pmatrix}
0 & 1 & 0 & 1 & 1 & 1\\
1 & 0 & 0 & 1 & 1 & 1\\
1 & 1 & 0 & 0 & 1 & 1\\
0 & 0 & 1 & 0 & 0 & 0\\
1 & 1 & 0 & 1 & 0 & 1\\
1 & 1 & 0 & 1 & 1 & 0\\
\end{pmatrix},
\end{equation}
\begin{equation}\label{eq.M12}
M_{4,5} = \begin{pmatrix}
0 & 1 & 1 & 0 & 1 & 1\\
1 & 0 & 1 & 0 & 1 & 1\\
1 & 1 & 0 & 0 & 1 & 1\\
0 & 0 & 0 & 1 & 0 & 0\\
1 & 1 & 1 & 0 & 0 & 1\\
1 & 1 & 1 & 0 & 1 & 0\\
\end{pmatrix}.
\end{equation}
Moreover, due to Eq.~\eqref{eq.cabxy} we have that $M_{j,i} = M_{i,j}^t$. Where we denote by $A^t$ the transpose of matrix $A$.
\begin{table}[]
    \centering
    \begin{tabular}{c|c|c|c}
        $c(a,b|x,y)$ &  $y=3$ & $y=4$ & $y=5$ \\
         \hline
      $x=3$   &  $\delta_{a,b}$ &  $M_{3,4}$ &  $M_{3,5}$ \\
      $x=4$   &  $M_{4,3}$ &  $\delta_{a,b}$ &  $M_{4,5}$ \\
      $x=5$   &  $M_{5,3}$ &  $M_{5,4}$ &  $\delta_{a,b}$ 
    \end{tabular}
    \caption{Combination of inputs used to prove that the colored odd-pointed star game does not define a facet of the local polytope.}
    \label{tab:game1}
\end{table} 

Our argument is based on the next observation. To saturate the local bound of Eq.~\eqref{eq.heptagram}, i.e., to reach $24$, a local strategy must only lose in one combination of inputs $(x,y)$. We can show that a local strategy that loses for $(x=3,y=4)$ cannot reach the local bound of $24$. In order to see this, first, let us consider the deterministic strategies $\delta_{a,f(x)}$ and $\delta_{b,g(y)}$ used by Alice and Bob, respectively. Where $f(x)$ and $g(y)$ are deterministic assignments of outputs for each input. Suppose that a local strategy $p^L(a,b|x,y)=\delta_{a,f(x)} \delta_{b,g(y)}$ loses only for the inputs $(x=3,y=4)$. This means that $\sum_{a,b} c(a,b|3,4)\delta_{a,f(3)} \delta_{b,g(4)} = 0$, or equivalently $c(f(3),g(4)|3,4)=0$. Moreover, since to win in $(x=3,y=3)$ and $(x=4,y=4)$ both deterministic strategies have to be the same, due to $c(a,b|z,z)=\delta_{a,b}$, we have that $f(3)=g(3)$ and $f(4)=g(4)$. This implies that this strategy also loses for $(x=4,y=3)$ due to
\begin{equation}
\begin{split}
    \sum_{a,b} c(a,b|4,3)\delta_{a,f(4)} \delta_{b,g(3)} &= c(f(4),g(3)|4,3)\\
     &= c(g(4),f(3)|4,3)\\
     &= 0,
\end{split}
\end{equation}
where in the last step we have used the symmetry $c(a,b|x,y) = c(b,a|y,x)$. This proves that any local strategy that loses for $(x,y) = (3,4)$ will also lose for $(x,y) = (4,3)$ reaching at most a value of 23. The same argument can be applied for any combination of inputs $x' \neq y'$, with $x'$ and $y'$ being the common lines of Alice and Bob. 

The above argument shows that among all the local strategies that reach $24$, none of them will lose in $(x,y) = (3,4),(4,3),(3,5),(5,3),(4,5),(5,4)$. As a consequence a linear combination of all the local points that reach 24 would create a vector that has entries equal to zero whenever $c(a,b|x' \neq y')=0$. This is already an indication that the saturating points are contained in a subspace of relatively low dimension. Next we show that this subspace is of a lower dimension than the necessary to be a facet. 

For a Bell inequality to define a facet of the local polytope, the set of points that saturate its local bound should span an affine subspace of dimension $dim(NS)-1$. $dim(NS)$ is the dimension of the no-signaling space and it is given by $dim(NS) = |X||Y|(|A|-1)(|B|-1) + |X|(|A|-1) + |Y|(|B|-1)$. For our case, $(|X|=|Y|=5,~|A|=|B|=6)$, we have $dim(NS) = 675$. As it can be noted, so far we have parameterized the correlations $\Vec{p}=\{ p(a,b|x,y)\}_{abxy}$ using a space of dimension $|X||Y||A||B|$, which is larger than $dim(NS)$. In order to represent the correlations in a space of dimension $dim(NS)$ we can use the parametrization introduced in \cite{CG2004JoP}, also known as Collins-Gisin (CG) notation. In this notation the correlations $\Vec{p}$ are parameterized using $dim(NS)$ probabilities such that $\Vec{p} = \{p(a,b|x,y), p^A(a|x), p^B(b|y)\}$, where $p^A$ ($p^B$) denotes the marginal probability of Alice (Bob). And, we also have that $a = 1,\dots, |A|-1$ and $x = 1,\dots, |X|$, while $b = 1,\dots, |B|-1$ and $y = 1,\dots, |Y|$. The CG notation is useful here because once we have removed the redundancy on the dimension of the correlation space we can directly translate the number of zeros of a linear combination of vectors with the dimension of its orthogonal complement.

Now the joint part of the correlations $p(a,b|x,y)$ will be described by a $5 \times 5$ matrix, instead of $6 \times 6$, for every combination of inputs $(x,y)$. Therefore we should only consider the first $5 \times 5$ entries in Eqs.(\ref{eq.M01}-\ref{eq.M12}), or in other words ignore the last row and the last column of these matrices. Then, a linear combination of all the saturating points will have: 12 zeros coming from $M_{3,4}$, 12 zeros from $M_{3,5}$, and 12 zeros from $M_{4,5}$, which including the coefficients in the transposition $M^t_{x,y}$ reaches a total of 72 zeros. Therefore in the most favorable case the maximum dimension of the hyperplane defined by Eq.~\eqref{eq.heptagram} would be $dim(NS)-72$. This upper bound on the dimension of the hyperplane defined by Eq.~\eqref{eq.heptagram} shows that it is not a facet.

Now we extend this proof to the family of colored odd-pointed star games. Clearly for a colored odd-pointed star game the number of common lines is larger. But, since we only aim at proving that the dimension of the hyperplane is lower than $dim(NS)-1$, we can focus only on two of the analogous lines we have already analyzed in the case of the $J(7,2)$-configuration $x=y=3,4$. These lines in general would be defined as:
\begin{equation}
\begin{split}
    x=y=3:&~\{13,23,34,\dots,3(N-1),3N\},\\
    x=y=4:&~\{14,24,34,\dots,4(N-1),4N\}.\\
\end{split}
\end{equation}
The analogous matrix to $M_{3,4}$ would have $3(N-2)$ zeros. The zeros that will be removed are the ones placed in $(a,b) = (N-1,1),(N-1,N-1)$, and $(3,N-1)$. Therefore, a linear combination of saturating points would have at least $3(N-3)$ zeros that correspond to the zeros in $(x,y) = (3,4)$. Thus the total number of zeros, including the transposition $(x,y) = (4,3)$, will be at least $6(N-3)$. Consequently, the dimension of the hyperplane will be lower than $dim(NS)-6(N-3)$, with $N\geq 7$. Hence Eq.~\eqref{eq.Ngram} does not define a facet.


\section{Advantages and weak points} \label{sec.Comparison}


Here, we summarize the advantages of the colored odd-star PQSs with respect to other bipartite PQSs.


\subsection{Bipartite PQSs based on qudit-qudit maximally entangled states and KS sets}


Existing bipartite PQSs based on qudit-qudit maximally entangled states and KS sets \cite{Stairs:1983PS,HR83,Brown:1990FPH,Elby:1992PLA,Renner2004b,CHTW04,Aolita:2012PRA} are constructed by one of two methods: 
either both Alice and Bob's measurement settings are all the bases of the corresponding KS set or the settings of one of the parties are all the bases while the settings of the other party are all the vectors.  This implies that the resulting bipartite PQSs largely fail to have a ``small'' input cardinality,
since, e.g., the simplest known KS sets in $d=3$ are the the Conway-Kochen $31$-vector KS set \cite{Peres:1993}, which has $17$ orthogonal bases, and the Peres $33$-vector set \cite{Peres:1991JPA}, which has $16$ orthogonal bases. In contrast, each of the KS sets that we are using have just $d+1$ orthogonal bases and the number of local settings of our PQSs is $d-1$. 


\subsection{Bipartite PQSs based on magic sets}


A simple way to achieve bipartite PQS with small input cardinality and high local dimension $d$ is to consider parallel repetitions \cite{Coladangelo2016arxiv,Coudron2016arxiv,Araujo:2020Quantum} of bipartite PQS of small input cardinality \cite{Cabello:2001PRLb,Mancinska:2007} based on magic sets.
For example, $N$ parallel repetitions of the magic square produce bipartite PQS with $d=4^N$, using only $3^N$ settings. Compared to this, the advantages of the approach presented in this paper are the following: (i) It generates bipartite PQSs with small input cardinality in every dimension $d = 2^kp^m$ with $m \geq k \geq 1$ for $p$ prime, $d =8p$ for $p>19$ prime, and also in other sporadic cases, rather than only when $d=(2^n)^N$, with $n \ge 2$. (ii) It produces ``genuinely'' $d$-dimensional PQSs in the sense that they cannot be achieved as parallel repetitions of the same smaller dimensional PQSs.


\subsection{Bipartite PQSs based on state-independent contextuality sets of vectors}


Existing bipartite PQSs based 
on state-independent contextuality sets of vectors whose corresponding orthogonality graph has fractional packing number equal to the Lov\'asz number and larger than the independence number, specifically those corresponding to
Pauli and Newman states \cite{ZSSLPC2023}, have two disadvantages: (i')~They only exist in specific dimensions. (i'')~They have a very large input cardinality (see \cite{xu2023stateindependent} for details).


\subsection{Noise} \label{sec.QSWPV}


The arguably main problem of {\em all} existing bipartite PQSs is their high sensitive to noise. 
Any noise rapidly pushes us out of perfection. Not only that: even small noise rapidly kills the quantum advantage. 

Here, we first compute the visibility (as defined below) required to beat the classical winning probability of the PQSs for the colored odd-star games. As we will see, it is very high. Then, we present two ways to mitigate this problem.


\subsubsection{Visibility for the colored odd-star game}


Suppose that the initial state is not the required maximally entangled state in local dimension $d$ but, instead, the Werner state
\begin{equation*}
    \rho(V_d) = V_d \ket{\psi}\bra{\psi} + \frac{1-V_d}{d^2}\mathds{1}_{d^2},
\end{equation*}
where $\mathds{1}_{d^2}$ denotes the $d^2 \times d^2$ identity matrix, $0 \leq V_d \leq 1$ is the visibility, and
\begin{equation}
    \ket{\psi} = \frac{1}{\sqrt{d}} \sum\limits_{k=0}^{d-1} \ket{k} \otimes \ket{k}.
\end{equation}

Then, if Alice measures $A_v = \ket{v}\bra{v} \otimes \mathds{1}_d$, and Bob measures $B_w = \mathds{1}_d \otimes \ket{w}\bra{w}$, the probabilities of the four possible outcomes are
{\tiny{
\begin{subequations}
    \begin{align}
    P(A_v=1, B_w=1 | \rho) &= \tr[\left(\ket{v}\bra{v} \otimes \ket{w}\bra{w}\right)\rho], \\
    P(A_v=1, B_w=0 | \rho) &= \tr[\left(\ket{v}\bra{v} \otimes (\mathds{1}_{d} - \ket{w}\bra{w})\right)\rho],   \\
    P(A_v=0, B_w=1 | \rho) &= \tr[\left((\mathds{1}_{d} - \ket{v}\bra{v}) \otimes \ket{w}\bra{w}\right)\rho], \\
    P(A_v=0, B_w=0 | \rho)  &= \tr[\left((\mathds{1}_{d} - \ket{v}\bra{v}) \otimes (\mathds{1}_{d} - \ket{w}\bra{w})\right)\rho].   
\end{align}
\end{subequations}
}}
In order to simplify the computation, we note that this computation is only necessary for our purposes when $v = w$ or when $v$ and $w$ are orthogonal. Thus, by symmetry, we only need to do computations using the two vectors $\ket{s}=\ket{0}^{\otimes d}$ and $\ket{t}=\ket{0}^{\otimes d-1} \otimes \ket{1}$, so that the above probabilities need only be computed with $v=s, w=s$ and with $v=s, w=t$. We now derive one of the eight possibilities by hand, noting that the others follow in a similar manner:
\begin{equation}
\begin{split}
    &P(A_s=1, B_s=1| \rho) \\
    =&\tr\left[\left(\ket{s}\bra{s} \otimes \ket{s}\bra{s}\right)\rho\right] \\
   =&V_d\tr\left[\left(\ket{s}\bra{s} \otimes \ket{s}\bra{s}\right)\ket{\psi}\bra{\psi}\right] \\
    &+\frac{1-V_d}{d^2}\tr\left[\left(\ket{s}\bra{s} \otimes \ket{s}\bra{s}\right)\right] \\
    =&\frac{(d-1)V_d + 1}{d^2}.
\end{split}
\end{equation}

Using this (and appropriate expressions for the other probabilities), we compute the quantum winning probability. As in the computation of the Bell inequality (\ref{eq.heptagram}), we consider the two cases:
\begin{enumerate}
    \item Alice and Bob receive the same line as input,
    \item Alice and Bob receive different lines as input.
\end{enumerate}

In Case (1), Alice and Bob win exactly when they make the same choice of point from the $d$ choices. This is equal to $d \cdot P(A_s=1, B_s=1 | \rho)$, which is equal to
\begin{equation}
    P_1 := \frac{(d-1)V_d+1}{d}.
\end{equation}

Now let us consider Case (2). Here, we compute the probability of failure. This happens in two different ways: either (2a) Alice and Bob disagree on their shared point (one of Alice or Bob outputs the intersection point of their two lines, and the other does not) or (2b) Alice and Bob choose points that are in a common (third) line.

The probability of failing in Case (2a) is 
\begin{equation}
\begin{split}
    &P(A_s=1, B_s=0|\rho) + P(A_s=0, B_s=1|\rho) \\
 = &\frac{2(d-1)(1-V_d)}{d^2}.
\end{split}
\end{equation}
The probability of failing in Case (2b) is the sum of probabilities that Alice and Bob measure orthogonal vectors over each of the $d-1$ pairs of points that occur in a common third line (for example, if Alice is given Line $1$ of the $J(7,2)$-configuration and Bob is given Line $2$ of the $J(7,2)$-configuration, then these pairs of points would be $(13,23), (14,24), (15,25),(16,26),(17,27)$). Thus, the probability of failing in Case (2b) is
\begin{align}
    &(d-1)P(A_s=1, B_t=1|\rho) \\
    = & \frac{(d-1)(1-V_d)}{d^2}
\end{align}
remarking again that $\ket{s}$ and $\ket{t}$ are orthogonal vectors.
Therefore, the probability of success in Case 2 is 
\begin{equation}
    P_2 := 1 - \frac{3(d-1)(1-V_d)}{d^2}.
\end{equation}
Finally, Alice and Bob have $d-1$ lines each, $d-3$ of which are in common, and so there are:
\begin{itemize}
    \item $d-3$ input choices where Alice and Bob share a common line,
    \item $(d-1)^2 - (d-3)$ input choices where Alice and Bob are given different lines.
\end{itemize}
Assuming that any pair of lines is equally likely to be chosen, we find that quantum winning probability is
\begin{equation}
    W := \frac{(d-3)P_1 + \left((d-1)^2 - (d-3)\right)P_2}{(d-1)^2}.
\end{equation}

On the other hand, there is a classical strategy that wins the game with probability $1 - \frac{1}{(d-1)^2}$, losing a single time out of the possible $(d-1)^2$ outcomes. 

In order to simplify our explanation, we assume that Alice may be given all of the lines except for $N-3$ and $N-2$ as input and that Bob may be given any of the lines except for $N-1$ and $N$. For each odd $1 \leq i \leq N-1$, Alice and Bob choose the point $\{i,i+1\}$ whenever either of them are given Line $i$ or $i+1$ as input. Finally, if Alice is given Line $N$, then she chooses the point $\{N-1, N\}$.

One may check that with the above strategy Alice and Bob only fail if Alice is given Line $N$ and Bob is given Line $N-2$ since Alice then chooses $\{N-1, N\}$ and Bob chooses $\{N-2, N-1\}$. These both lie on the common line $N-1$, and so Alice and Bob lose in this case. By the Bell inequality (\ref{eq.Ngram}), it follows that this classical strategy is optimal. 

Therefore, in order for the quantum strategy to improve on the classical strategy, we must have 
\begin{equation}
    W \geq 1 - \frac{1}{(d-1)^2},
\end{equation}
which is accomplished exactly when
\begin{equation}
    V_d \geq \frac{(d-2)(4d^2-9d+6)}{4(d-1)(d^2 - 3d + 3)}.
\end{equation}
In the cases of $d = 6,8,10$, we get $32/35 \approx 0.914$, $285/301 \approx 0.947$, and $632/657 \approx 0.962$, respectively for the RHS, and for $d$ large, the RHS approaches $1$.


\subsubsection{Visibility for the non-colored odd-star line-line game}


In this subsection and the following we compare the colored odd-pointed star game against two natural games which could be played on the star. In the current subsection, we consider the same game that we have described, except that both Alice and Bob may each receive any of the $N$ lines as input. We call this the \emph{non-colored odd-star line-line game}.

Here, Alice and Bob can win with probability $1-\frac{4}{(d+1)^2}$ by using the following rule when receiving Line $i$ as input:
\begin{itemize}
    \item if $i$ is odd and $i \neq N$, then they both choose point $\{i,i+1\}$,
    \item if $i$ is even or $i=N$, then they both choose point $\{i, i-1\}$.
\end{itemize}
With this strategy Alice and Bob lose exactly in the four cases when one of them is given Line $N-2$ or $N-1$ and the other is given Line $N$. They win in all other cases. 

In order to show that this strategy is optimal, we consider two cases:
\begin{itemize}
    \item when Alice and Bob use the same strategy (i.e., choose the same point for each line),
    \item when Alice and Bob do not use the same strategy (i.e., there is at least one line for which Alice and Bob choose different points).
\end{itemize}
Let us consider the first case. Since $N$ is odd, the set $\{1,2,\dots,N\}$ cannot be partitioned into sets of size $2$. Therefore, we see that there must be lines $i,j$ for which Alice and Bob choose the point $\{i,j\}$ in Line $i$, but do not choose the point $\{i,j\}$ in Line $j$. Let us say that Alice and Bob instead choose point $\{j,k\}$ in Line $j$ for some $k \neq i$. We immediately see that Alice and Bob lose if Alice is asked Line $i$ and Bob is asked Line $j$ and vice-versa. Now consider Line $k$. No matter what point we choose for Line $k$, we see that Alice and Bob must fail in at least one of the following two scenarios:
\begin{itemize}
    \item Alice is given Line $i$ and Bob is given Line $k$ (and vice-versa),
    \item Alice is given Line $j$ and Bob is given Line $k$ (and vice-versa).
\end{itemize}
Therefore, we see that Alice and Bob must indeed fail at least 4 times in any shared strategy.


Now let us consider the second case. Here Alice and Bob disagree on some Line $i$, and so must fail if both asked for this line. Without loss of generality we may assume that the line is Line~$1$, and that Alice chooses $12$ and Bob chooses $13$ when both are given this line. The following fact allows us to simplify matters significantly:
\begin{proposition}
    Assume that Alice chooses $12$ and Bob chooses $13$ when given Line $1$ as input. 
    Let $2 \leq i \leq N$. For Line $i$ if Alice or Bob make any choice from the points $\{1,i\}, \{2,i\}, \{3,i\}$, then Alice and Bob will fail for
    some choice of lines which includes Line $i$ and Line $j$ for some $j \in \{1,i\}$.
\end{proposition}

\begin{proof}
    If Alice and Bob choose different points when both given Line $i$, then they will fail in this case, therefore we may assume that they choose the same point $\{i,k\}$ with $k \in \{1,2,3\}$. However, one of the intersections $\{i,k\} \cap \{1,2\}$ or $\{i,k\} \cap \{1,3\}$ must have size $1$. Therefore, we see that there is indeed a failure involving Line $i$, and that the other line is either Line $1$ or Line $i$.
\end{proof}

Therefore, we see that Alice and Bob must fail at least three times: once when both given Line $1$, once for some choice including Line $2$ and Line $1$ or $2$, and once for some choice including Line $3$ and Line $1$ or $3$.

Clearly, then, in order to have fewer than $4$ failures, Alice and Bob must use a shared strategy for the lines $\{4,\dots, N\}$, and for each of these lines they must not choose a point that intersects $\{1,2,3\}$. By the same token, for each Line $i \in \{1,2,3\}$, any choice of point $\{i,j\}$ with $j \notin \{1,2,3\}$ for either Alice or Bob would lead to a failure when the other is given Line $j$. Therefore, we see that each choice of both Alice and Bob from lines $1,2,3$ must be a subset of $\{1,2,3\}$. Then we immediately have three failures: when Alice and Bob are both given Line $1$, when Alice is given Line $1$ and Bob is given Line $3$, and when Alice is given Line $2$ and Bob is given Line~$1$. Therefore, they must use a shared strategy for Lines $2$ and $3$ in order to avoid a fourth failure. One may quickly check by brute force that any such strategy does lead to a fourth failure. Finally, we conclude that an optimal classical strategy for Alice and Bob is the aforementioned shared strategy leading to exactly $4$ failures.

On the other hand, we find that the quantum winning probability is 
\begin{equation}
    \frac{(d+1)P_1 + ((d+1)^2-(d+1))P_2}{(d+1)^2},
\end{equation}
which is higher than the classical winning probability when
\begin{equation}
    V_d \geq \frac{d^2 - d - 1}{(d+1)(d-1)}.
\end{equation}
For $d=6,8,10$, we thus need $V_d$ to be at least $29/35 \approx 0.829$, $55/63 \approx 0.873$, and $89/99 \approx 0.899$, respectively. 


\subsubsection{Visibility for the non-colored odd-star point-line game}


Here we consider the classical KS game, where Alice may receive any of the $N$ lines as input, and Bob may receive any of the $N-1$ points on Alice's line as input. In this case Alice chooses any of the $N-1$ points on the line and indicates this by labeling the point with a $1$, and labeling the remaining points on the line with $0$ (therefore Alice has $N-1$ choices of output). Bob then labels his point with either $0$ or $1$ (therefore Bob has $2$ choices of output). They win if and only if they agree on Bob's point. We call this instance of the classical KS game the \emph{non-colored odd-star point-line game}.

In the classical KS game, Alice and Bob can win with probability $1 - \frac{1}{d(d+1)}$ as follows. For each odd $1 \leq i < N$, Alice labels with $1$ the point $\{i,i+1\}$ with $1$ for Line $i$ and Line $i+1$. Then, for the final line, Alice can choose any point to label with $1$. With the exception of the final point, Bob also labels with a $1$ each of Alice's points that are labeled with $1$. Bob labels the final point with a $0$. Then the only losing combination is if Alice is given Line $N$ and Bob is given Alice's final point. This strategy is optimal since there is no perfect classical strategy. Here, the quantum value is easy to compute: it is given by
\begin{align*}
    &P(A=1, B=1 | \rho) + P(A=0, B=0 | \rho) \\
    = &\frac{d^2 - 2(d-1)(1-V_d)}{d^2}
\end{align*}
when Alice and Bob both use $v$ to measure, and so in this case, we find that the quantum winning probability is higher than the classical winning probability when
\begin{equation}
    V_d \geq \frac{2d^2-d-2}{2(d+1)(d-1)},
\end{equation}
which, for $d=6,8,10$, yields $32/35 \approx 0.914$, $59/63 \approx 0.937$, and $94/99 \approx 0.949$, respectively.


\section{Summary and open problems}
\label{sec:conc}


In this article, we have introduced two related results. On the one hand, we have extended a family of KS sets in dimension $d$, with $d+1$ orthogonal bases. Now the family covers any:
\begin{itemize}
    \item $d = 2^kp^m$ for $p$ prime, $m \geq k \geq 0$ (with the exception of $k=m=0$),
    \item $d = 8p$ for $p \geq 19$,
    \item $d = pk$ whenever there exists a classical Hadamard matrix of order $k$ and we have $p > ((k-2)2^{k-2})^2$,
    \item any $d$ large enough, satisfying the congruence $d+1 \equiv d'+1 \mod (d'+1)(d'))$ for some $d'$ already covered by the family. For example, any $d > 294427, d \equiv 7 \mod 42$,
    \item any dimension $d$ that may be constructed using \cite[Theorems 5.11 and 5.12]{Colbourn:1995},
    \item any $d$ satisfying $d+1 = (d'+1)^m$ where $d'$ is already covered, $d'+1$ is a prime power, and $m$ is any positive integer,
    \item $d=8$,
    \item other sporadic examples.
\end{itemize}
 On the other hand, we have shown how to distribute each of these KS sets between two parties and produce a PQS in such a way that the product of the number of settings of Alice times the number of settings of Bob is minimum. As a result, we have presented a family of PQSs in any $(2,d-1,d)$ Bell scenario. The KS sets are the ones having the smallest number of bases known. For dimensions different than $4^N$ and $8^N$ (which correspond to the parallel version of the games based on magic sets for two and three qubits), the corresponding PQSs are the ones with the smallest number of settings known.

Here, we list here some problems for future research:
\begin{itemize}
 \item Find ORs of dimension $d$ for the $J(d+1,2)$ for the even dimensions not covered by the results of this paper. For $d = 4$, it is easy to see that such ORs do not exist \cite{LisonekPRA2014}. The first case that remains open is $d=32$. It has been conjectured that for each even dimension $d$ and $p$ prime dividing $d$, there is a generalized Hadamard matrix $GH(p,\lambda)$ with $p\lambda = d$ \cite[Conjecture 5.18 (3)]{Colbourn:1995}. This would resolve each case when the dimension that is not a power of $2$ (since Lison\v ek's construction does not work for $p=2$). The smallest open case is dimension $30$ with $p=5$.
\item Identify $(2,d-1,d)$ PQSs for {\em odd} $d \ge 5$. For, $d=3$, it can be proven that $(2,2,3)$ PQSs do not exist \cite{Gisin:2007IJQI}. Moreover, the simplest qutrit-qutrit PQS known requires $9$ settings in one party and $7$ settings in the other \cite{Cabello:2023XXX}.
\item An open problem is, given $d$, what is the bipartite PQS with the smallest input cardinality. For $d=3$ and $4$, it has been conjectured that the solution is $(2,9-7,3)$ and $(2,3,4)$, respectively \cite{Cabello:2023XXX}. For $d=4^N$, the PQSs with the smallest known input cardinality, $3^N$ for both parties, are achieved by taking $N$ parallel repetitions of the magic square correlations. For $d=8^N$, the parallelization of the Mermin star \cite{Mermin:1990PRLb} game due to Mancinska \cite{Mancinska} has input cardinality $5^N$ for both parties. In fact, this may be reduced to $5^N$ for Alice and $4^N$ for Bob by noting that in the original formulation of the game, one can remove one of the inputs from either player while still not allowing a perfect classical strategy. This seems to be the smallest known input cardinality in dimension $8^N$.
Of course, there are ``trivial'' PQSs in any $(2,n \ge 3,d \ge 4)$ Bell scenario, so the question is rather what is the bipartite PQS with the smallest input cardinality for the corresponding $d$ in the sense that it cannot be traced back to a PQS of smallest $d$ as in, e.g., \cite{Arkhipov:2012XXX}, or in the previous examples. 

 \item We noted in Sec.~\ref{sec.Optimality} that the colored odd-pointed star game does not depend on the KS set given, only on the orthogonalities of the $J(N,2)$ graph. 
 In principle, by taking account the additional orthogonalities of the KS set, the classical winning probability could potentially be lower. More formally, for some KS set $\mathcal{K} = (\mathcal{V}, \mathcal{B})$, one can play the following variation of the colored odd-pointed star game. Give Alice an orthogonal basis $b \in S_A$ and Bob an orthogonal basis $b' \in S_B$ for bipartite KS $(S_A,S_B)$. Alice and Bob choose vectors $v \in b, v' \in b'$. They win if and only if $v$ and $v'$ are not orthogonal. Investigating these variations may be of interest.
    \item For the KS sets of Sec.~\ref{sec.InfiniteClass}, the classical value of this game (with $S_A, S_B$ chosen as in the colored odd-pointed star game) is at most the classical value of the colored odd-pointed star game. However, even with the additional orthogonalities, the classical value may be the same. It is then natural to ask if there is
good way to modify the KS sets of Sec.~\ref{sec.InfiniteClass} (by adding additional vectors) in order to ensure that the classical value is lowered. Such constructions would increase the resistance to noise. This connects to the most general open problem of constructing bipartite PQSs with larger resistance to noise.
\end{itemize}


\section*{Acknowledgments}


{This work was supported by EU-funded project \href{10.3030/101070558}{FoQaCiA}. A.C.\ acknowledges support from the \href{10.13039/501100011033}{MCINN/AEI} (Project No.\ PID2020-113738GB-I00).
}

\bibliographystyle{quantum}

\appendix


\section{FN, FNS, AVN} \label{app.Correlationstrength}


Here, we present the definitions of FN correlations, FNS correlations, and AVN proofs of Bell theorem.

Consider the $(|X|, |A|; |Y|, |B|)$ Bell scenario, where Alice has $|X|$ settings with $|A|$ outcomes, and Bob has $|Y|$ settings with $|B|$ outcomes. A correlation $p(a,b|x,y)$ may be viewed as a point in $\mathbb{R}^{|A||B||X||Y|}$. For any choice of $A,B,X,Y$, we obtain different sets of correlations by considering different theories.

A correlation $p(a,b|x,y)$ is  \emph{deterministic} if each output is a function of the corresponding input, i.e., $a = f(x), b = g(y)$ for some deterministic functions $f,g$. A \emph{local correlation} is any correlation that can be written as convex combination of deterministic correlations. Therefore, the set of local correlations forms a polytope, called the \emph{local polytope}, whose vertices are  the deterministic correlations.

A correlation $p(a,b|x,y)$ is {\em quantum} if it can be realized as
\begin{equation}
    p(a,b|x,y) = \tr[(M_{a|x} \otimes M_{b|y})\rho]
\end{equation}
for some quantum state $\rho$ and (POVM) measurements $M_{a|x}$ and $M_{b|y}$. 

A correlation $p(a,b|x,y)$ is \emph{non-signaling} (NS) if, for every $a,x,y,y'$,
\begin{equation}
    \sum_{b} p(a,b | x,y) = \sum_{b} p(a,b|x,y')
\end{equation}
and, for every $b,x,x',y$,
\begin{equation}
    \sum_{a} p(a,b | x,y) = \sum_{a} p(a,b|x',y).
\end{equation}
The set of all NS correlations is a polytope, the \emph{NS polytope}. 

The local polytope is a subset of the set of quantum correlations, which in turn, is a convex subset of the NS polytope \cite{Brunner:2014RMP}.


{\em Full nonlocality:} Any NS correlation $p(a,b|x,y)$ can be decomposed as
\begin{equation}
\label{deco}
p(a,b|x,y) = q_{\mathrm{L}} p_{\mathrm{L}}(a,b|x,y) + (1 - q_{\mathrm{L}})p_{\mathrm{NL}}(a,b|x,y),
\end{equation}
where $p_{\mathrm{L}}(a,b|x,y)$ is a local correlation and $p_{\mathrm{NL}}(a,b|x,y)$ is a nonlocal NS correlation, with respective weights $q_{\mathrm{L}}$ (called the \emph{local weight}) and $1 - q_{\mathrm{L}}$ (called the \emph{nonlocal weight}), for some $0\leq q_{\mathrm{L}} \leq 1$. The \emph{local content} or \emph{local fraction} $q_{\mathrm{L}}^{\mathrm{max}}$ of $p(a,b|x,y)$ is the largest possible local weight over all such decompositions \cite{Elitzur:1992PLA}. Formally, $q_{\mathrm{L}}^{\mathrm{max}} := \max_{\{p_{\mathrm{L}},p_{\mathrm{NL}}\}} q_{\mathrm{L}}$. Conversely, the \emph{nonlocal content} is the smallest possible nonlocal weight -- thus, $q_{\mathrm{NL}}^{\mathrm{min}} := 1-q_{\mathrm{L}}^{\mathrm{max}}$.
The correlation $p(a,b|x,y)$ is local if and only if the nonlocal content is zero (i.e., $q_{\mathrm{NL}}^{\mathrm{min}}=0$). 
A correlation $p(a,b|x,y)$ is {\em fully nonlocal} (FN) \cite{Aolita:2012PRA} or {\em strongly nonlocal} \cite{Abramsky2019Pro} whenever $q_{\mathrm{NL}}^{\mathrm{min}}=1$.


{\em Face nonsignaling correlations:} A nonlocal correlation $p(a,b|x,y)$ is a point in the NS polytope that lies outside of the local polytope \cite{Pitowsky:1989}. A nonlocal correlation $p(a,b|x,y)$ is \emph{face nonsignaling} (FNS) if it is on a face of the NS polytope containing no local points.


{\em All-versus-nothing proofs:} 
The \emph{support} $s$ of a correlation $p(a,b|x,y)$ is the set of indices for which $p(a,b|x,y)$ is non-zero. An \emph{all-versus-nothing (AVN) proof} of Bell theorem \cite{Mermin:1990PRLa} is a quantum correlation for which there is no local correlation whose support is a subset of $s$.


\section{Connection between perfect quantum strategies and Kochen-Specker sets} \label{app.connection}


Here, we summarize the one-to-one connection between PQSs and KS sets proven in \cite{Cabello:2023XXX}. The definition of KS set is in Sec.~\ref{sec.N=7} (see Definition~\ref{def:ks}). The definition of B-KS set is in Sec.~\ref{sec.Optimality} (see Definition~\ref{def:bks}).

Every B-KS set defines a PQS for the following bipartite game: Alice (Bob) receives an orthogonal basis $x$ ($y$) chosen at random from $S_A$ ($S_B$) in Definition~\ref{def:bks}. Alice (Bob) outputs a vector of $x$ ($y$). Alice and Bob win if and only if they output non-orthogonal vectors. 
The PQS is realized as follows: Alice and Bob share a qudit-qudit maximally entangled state in dimension $d$. For each $x$ ($y$), Alice (Bob) measures the orthogonal projectors onto the vectors $\{s_{a|x}\}_a$ ($\{s_{b|y}\}_b$) and outputs the vector associated to the result that she (he) has obtained.

Reciprocally, every bipartite PQS defines a B-KS set. 
Every PQS consists of a quantum correlation. As shown in, e.g., \cite{Brunner:2014RMP,Paddock:2022XXX}, every quantum correlation can be realized as
\begin{equation}
    p(a,b|x,y) = \langle \psi | \Pi_{a|x} \otimes \Pi_{b|y} | \psi \rangle,
\end{equation}
where $\ket{\psi} \in \mathcal{H} = \mathcal{H}_A \otimes \mathcal{H}_B$ is a pure state, $\Pi_{a|x}$ are projective measurements for Alice (i.e., $\Pi_{a|x} \Pi_{a'|x}= \delta_{a a'} \Pi_{a|x}$ and $\sum_{a} \Pi_{a|x} = I_A$, where $I_A$ is the identity matrix in $\mathcal{H}_A$), and $\Pi_{b|y}$ are projective measurements for Bob (i.e., $\Pi_{b|y} \Pi_{b'|y}= \delta_{b b'} \Pi_{b|y}$ and $\sum_{b} \Pi_{b|y} = I_B$, where $I_B$ is the identity matrix in $\mathcal{H}_B$). 
For each $a|x$, we define 
\begin{equation}
    \ket{\psi_{a|x}} :=\frac{(\Pi_{a|x} \otimes I_B) |\psi\rangle}{\sqrt{\langle \psi | (\Pi_{a|x} \otimes I_B) | \psi \rangle}}.
\end{equation}
Similarly, for each $b|y$, we define 
\begin{equation}
    \ket{\psi_{b|y}} :=\frac{(I_A \otimes \Pi_{b|y}) |\psi\rangle}{\sqrt{\langle \psi | (I_A \otimes \Pi_{b|y}) | \psi \rangle}}.
\end{equation}
Then, the following result is proven in \cite{Cabello:2023XXX}.

\begin{theorem} \label{mt}
    A correlation $p(a,b|x,y)$ allows for a PQS in the $(|X|, |A|; |Y|, |B|)$ Bell scenario if and only if there is a B-KS set $S=S_A \cup S_B$, where $S_A := \{\ket{\psi_{a|x}} : a \in A, x \in X\}$ and $S_B := \{\ket{\psi_{b|y}} : b \in B, y \in Y\}$.
\end{theorem}


\end{document}